\documentclass{LMCS}

\def\dOi{10(2:2)2014}
\lmcsheading%
{\dOi}
{1--23}
{}
{}
{Jul.~\phantom09, 2013}
{May\phantom.~13, 2014}
{}


\ACMCCS{[{\bf Theory of computation}]: Semantics and
  reasoning---Program semantics---Categorical Semantics;
  Logic---Linear logic}

\usepackage{hyperref}
\usepackage{amssymb}
\usepackage{bm}

\newif\ifUSEEXTERNALGRAPHICS
\USEEXTERNALGRAPHICStrue

\ifUSEEXTERNALGRAPHICS
  \usepackage{tikzexternal}
  \tikzsetexternalprefix{figures-poly_functors/}

\else
  \usepackage{tikz}
  \usetikzlibrary{matrix,arrows,calc}

  \usetikzlibrary{external}
  \tikzset{morphism/.style={->,font=\scriptsize}}
  \tikzset{injection/.style={right hook->,font=\scriptsize}}
  \tikzset{twocell/.style = {double equal sign distance,double,-implies,shorten <= .3cm,shorten >=.3cm,font=\scriptsize,draw}}
  \tikzset{crossover/.style={preaction={solid,draw=white,-,line width=6pt}}}
  \tikzset{over/.style={fill=white,inner sep=1.5pt}}

  \newcommand{\pullback}[5][.7cm]{%
    \node[coordinate] (a) at (#2) {} ; %
      \node[coordinate] (b) at (#3) {} ; %
      \node[coordinate] (c) at (#4) {} ; %
      \path let
      \p1= ($(a) - (b)$) , %
      \p2= ($(c) - (b)$), %
      \n1={veclen(\x1,\y1)}, %
      \n2={veclen(\x2,\y2)}, %
      \n3={min(#1, .4 * min(\n1,\n2))},%
      \p3=($(\x1/\n1,\y1/\n1)$),%
      \p4=($(\x2/\n2,\y2/\n2)$),%
      \p5=($(\n3 * \x3, \n3 * \y3)$),%
      \p6=($(\n3 * \x4, \n3 * \y4)$) in%
      (b) ++ (\p5)  node[coordinate] (x) {} 
      (b) ++ (\p6)  node[coordinate] (y) {} 
      ;
    \node[coordinate] (c) at (barycentric cs:x=1,y=1,b=-1) {} ; %
      \path (x) -- node[pos=.02] (xup){} (c) ; %
      \path (y) -- node[pos=.02] (yup){} (c) ; %
      \path[#5] (xup) -- (c) -- (yup) ; %
  }
\fi

\newcommand\inter[2]{[\![#2]\!]\,=\,#1}
\newcommand\Rule[3]{\frac{\Big.\quad{#1}\quad}{\Big.#2}\quad\hbox{\scriptsize #3}}

\newcommand\C{\mathbb{C}}
\newcommand\Set{\mathsf{Set}}
\newcommand\Span{\mathsf{Span}}
\newcommand\Poly{\mathsf{PolyDiag}}
\newcommand\PolyFun{\mathsf{PolyFun}}
\newcommand\PDSim{\mathsf{PDSim}}
\newcommand\PFSim{\mathsf{PFSim}}
\newcommand\FSim{\mathsf{FSim}}
\newcommand\Nat{\mathsf{Nat}}

\newcommand\Sl[2]{{#1}/_{\!#2}}       
\newcommand\sq{\mathbin{\scriptstyle\square}}
\newcommand\tr{\mathbin{\triangleright}}
\DeclareMathOperator{\Lan}{\mathrm{Lan}}
\DeclareMathOperator{\Ran}{\mathrm{Ran}}
\DeclareMathOperator{\refl}{\mathrm{refl}}
\DeclareMathOperator{\id}{id}
\DeclareMathOperator{\tw}{tw}

\newcommand\op{\mathrm{op}}
\newcommand\End[1]{\int_{#1}}
\newcommand\Coend[1]{\int^{#1}}
\newcommand\fun{\boldsymbol\lambda}

\newcommand\iso{\cong}

\newcommand\Zero{\mathbf{0}}
\newcommand\One{\mathbf{1}}
\newcommand\Plus{\oplus}
\newcommand\Tensor{\otimes}
\newcommand\Linear{\multimap}
\DeclareSymbolFont{fontexp}{OT1}{cmr}{bx}{n}
\DeclareMathSymbol{!}{\mathalpha}{fontexp}{`!}

\newcommand\Pow{\mathcal{P}}
\newcommand\Mf{\ensuremath{\mathcal{M}_{\!f}}}

\renewcommand\iff{\Leftrightarrow}
\newcommand\eqdef{\stackrel{\smash{\scriptscriptstyle\mathsf{def}}}{=}}

\newcommand\BLANK{{\texttt{\char"5F}}}

\newcommand\Id[2]{\mathrm{Id}\left(#1,#2\right)}
\newcommand\IdX[3]{\mathrm{Id}_{#1}\left(#2,#3\right)}
\newcommand\eq{\mathrm{eq}}

\newcommand\DU[1]{\boldsymbol{\left[}{#1}\boldsymbol{\right]}}
\newcommand\AU[1]{\boldsymbol{\left\langle}{#1}\boldsymbol{\right\rangle}}

\newcommand\D{\mathrm{D}}

\newcommand\Sem[1]{[\![#1]\!]}          

\begin{document}

\title[A Linear Category of Polynomial Functors]
      {A Linear Category of Polynomial Functors (extensional part)}

\author{Pierre Hyvernat}
\thanks{This work was partially funded by the French ANR project r\'ecr\'e ANR-11-BS02-0010.}
\address{Laboratoire de Math\'ematiques\\
  CNRS UMR 5126 -- Universit\'e de Savoie\\
  73376 Le Bourget-du-Lac Cedex\\
  France}
\email{pierre.hyvernat@univ-savoie.fr}
\urladdr{\url{http://lama.univ-savoie.fr/~hyvernat/}}


\date{\today}

\begin{abstract} 

  We construct a symmetric monoidal closed category of \emph{polynomial
  endofunctors} (as objects) and \emph{simulation cells} (as morphisms). This
  structure is defined using universal properties without reference to
  representing polynomial diagrams and is reminiscent of Day's convolution on
  presheaves. We then make this category into a model for intuitionistic linear
  logic by defining an additive and exponential structure.

\end{abstract} 

\maketitle


\section*{Introduction}  

Polynomial functors are (generalizations of) functors~$X\mapsto \sum_k C_k
X^{E_k}$ in the category of sets and functions. Both the
``coefficients''~$C_k$ and the ``exponents''~$E_k$ are sets; and sums,
products and exponentiations are to be interpreted as disjoint unions,
cartesian products and function spaces. All the natural parametrized algebraic
datatypes arising in programming can be expressed in this way. For example,
the following datatypes are polynomial:
\begin{itemize}
  \item $X\mapsto \mathrm{List}(X)$ for lists of elements of~$X$, whose
    polynomial is~$\mathrm{List}(X) = \sum_{n\in\mathbb{N}}X^{[n]}$
    where~$[n]=\{0,\dots,n-1\}$;

  \item $X\mapsto \mathrm{LBin}(X)$ for ``left-leaning'' binary trees with
    nodes in~$X$, whose polynomial can be written
    as~$\mathrm{LBin}(X)=\sum_{t\in T} X^{N(t)}$, where~$T$ is the set of
    \emph{unlabeled} left-leaning trees and $N(t)$ is the set of nodes of~$t$;

  \item $X\mapsto \mathrm{Term}_S(X)$ for well-formed terms built from a first-order
    multi-sorted signature~$S$ with variables of sort~$\tau$ taken
    in~$X_\tau$.
\end{itemize}
In the last example~$X$ is a family of
sets indexed by sorts rather than a single set, and expressing it as a
polynomial requires ``indexed'' or ``multi-variables'' polynomial functors.

Because of this, those functors have recently received a lot of attention from
a computer science point of view. In this context, they are often called
\emph{containers}~\cite{cont,indexed-containers} and coefficients and
exponents are called \emph{shapes} and \emph{positions}. An early use of them (with
yet another terminology) goes back to Petersson and Synek~\cite{treeset}:
\emph{tree-sets} are a generalization of so called W-types from dependent type
theory. They are inductively generated and are related to the free monads of
arbitrary polynomial functors. In the presence of extensional equality, they
can be encoded using usual W-types~\cite{polyMonads}.

Polynomial functors form the objects of a category with a very rich structure.
The objects (i.e., polynomial functors), the morphisms (called
\emph{simulations}) and many operations (coproduct, tensor, composition, etc.)
can be interpreted using a ``games'' intuition (refer
to~\cite{giovanni,polyDiagrams} for more details):
\begin{itemize}
  \item a polynomial is a two-players game, where moves of the first player are given by its
    ``coefficients'' and counter-moves of the opponent are given by its
    ``exponents'',
  \item a simulation between two such games is a witness for a kind of back-and-forth property
    between them.
\end{itemize}
This category has enough structure to model intuitionistic linear
logic~\cite{polyDiagrams}. The simplest way to define this structure is to use
representations of polynomial functors, called \emph{polynomial diagrams} as
the objects. The fact that such representations give rise to functors isn't
relevant! This paper gives a functorial counterpart: a model for
intuitionistic linear logic where formulas are interpreted by polynomial
\emph{functors}, without reference to their representations.

The difference between the two approaches is subtle. Differentiation of plain
polynomials over the real numbers provides a useful analogy. It can be defined
in two radically different ways:
\begin{enumerate}
  \item on \emph{representations} of polynomials with the formula:
    \[
      P = \sum_{k=0}^{n} a_k X^k
      \quad \mapsto \quad
      P' = \sum_{k=1}^{n} k a_{k} X^{k-1}\ ,
    \]

  \item on polynomials \emph{as functions} with the formula:
    \[
      F \quad \mapsto \quad
      F' = x \mapsto \lim_{\varepsilon \to 0} \frac{F(x+\varepsilon)-F(x)}{\varepsilon} \ .
    \]
\end{enumerate}
Both definitions are useful and neither is obviously reducible to the other.
The first one is easier to work with for concrete polynomials but the second
one applies to a larger class of functions.
The model of intuitionistic linear logic previously
defined~\cite{polyDiagrams} was ``intensional'': just like point~(1), it
used representations of polynomial functors.
The model described here is ``extensional'': just like point~(2), it
applies to arbitrary functors, even though it needs not be defined for non
polynomial ones.

More precisely, the category defined in~\cite{polyDiagrams} is equivalent
(Proposition~\ref{prop:PEequivPEFun}) to a subcategory of a category of
arbitrary functors with simulations, where the tensor and its adjoint are not
always defined.

\medbreak
We first (Section~\ref{section:prelim}) recall some notions about polynomial
functors. We then define the category of polynomial functors with simulations
and show that it is symmetric monoidal closed (Section~\ref{section:PE}). We
relate this structure to a generalized version of Day's convolution product on
presheaves in Section~\ref{sub:dayConvolution}. We finish by showing how the
additive and exponential structure from~\cite{polyDiagrams} can be
recovered.


\subsection*{Related Works} 

The starting point of this work is the characterization of strong natural
transformations between polynomial functors due to N.~Gambino and J.~Kock
\cite{polyMonads}.
However, our notion of morphism between polynomial functors is more general
than what appears in~\cite{notesJoachim,polyMonads}
or~\cite{cont,indexed-containers} (where polynomial functors are called
\emph{indexed containers}). In particular, there can be morphisms between
polynomial functors that do not share their domains and codomains.

Another inspiration is the model of ``predicate transformers''
from~\cite{denotPT}. This is a model for \emph{classical} linear logic where
formulas are interpreted by monotonic operators on subsets and proofs are
interpreted by ``simulations''. We show here that ``monotonic operator
on~$\Pow(I)$'' can be replaced by ``functor on~$\Sl\Set I$'' with the
following restrictions:
\begin{itemize}
  \item we give up classical logic;
  \item we only consider \emph{polynomial} functors.
\end{itemize}



\section{Preliminaries: Polynomial Functors}  
\label{section:prelim}

Basic knowledge about locally cartesian closed categories and their internal
language (extensional dependent type theory) is assumed throughout the paper.
Refer to Appendix~\ref{app:LCCC} and~\ref{app:DTT} for the relevant
definitions and notation. The first half of~\cite{polyMonads} is of crucial
importance and we start by recalling some results, referring to the original
article for details.
From now on, $\C$ will always denote a locally cartesian closed category.

\subsection{Polynomials and Polynomial Functors} 

A polynomial is, in the usual sense, a function of several variables that can
be written as a sum of monomials, where each monomial is a product of a
(constant) coefficient and several variables. A polynomial functor is similar,
with the following differences
\begin{itemize}
  \item it acts on sets rather than numbers;
  \item sums and products are the corresponding set-theoretic operations;
  \item it may have arbitrarily many arguments: instead of a
    tuple~$(X_1,\dots,X_n)$ of variables, it acts on \emph{families}~$(X_i)_{i\in I}$, for
    a given set~$I$;
  \item the sum of monomial isn't necessarily finite and can be 
    indexed by any set;
  \item because~$C\times Y \iso \sum_{c\in C} Y$, we don't need
    constant coefficients in monomials;
  \item each monomial is an indexed product of variables~$X_i$.
\end{itemize}
A single polynomial functor is thus made up from the following data:
\begin{itemize}
  \item a set~$I$ indexing the variables,

  \item a set~$A$ indexing the sum of monomials,

  \item an~$A$-indexed family of sets~$(D_v)_{v\in A}$ where for each~$v\in
    A$, i.e., for each monomial appearing in the sum, the set~$D_v$ indexes
    the variables composing the monomial,

  \item for each~$v\in A$, a function~$D_v \to I$ giving the indices of the
    variables of the monomial. For example, if the monomial uses a single
    variable, this function is constant...
\end{itemize}
The corresponding functor is given by
\[
  \big(X_i\big)_{i\in I} \quad \mapsto \quad \sum_{v\in A} \ \prod_{u\in D_v} X_{n(u)}
  \ .
\]
Categorically speaking, a~$K$-indexed family of sets can be represented by a function
with codomain~$K$. If~$\gamma : S \to K$, the associated family~$\Gamma$ will be given
by~$\Gamma_k = \gamma^{-1}(k)$. A polynomial functor can thus be represented
with
\begin{itemize}
  \item a set~$I$ indexing the variables;
  \item a set~$A$ indexing the sum;
  \item a function~$d:D\to A$ giving an~$A$-indexed family of sets indexing
    the products of variables;
  \item a function~$n:D\to I$ giving, for each variable in each product, its
    index.
\end{itemize}
Written in full, the corresponding functor is
\[
  \big(X_i\big)_{i\in I} \quad \mapsto \quad \sum_{v\in A} \ \prod_{u\in
  d^{-1}(v)} X_{n(u)}
  \ .
\]
In order to compose such functors, we need to consider functors acting on
families of sets, and giving families of sets as a result. We can thus have
polynomial functors taking~$I$-indexed families and giving~$J$-indexed
families of sets. Such a functor is simply a~$J$-indexed family of functors
in the above sense and is thus of the form:
\begin{equation}\label{eqn:ext_in_set}
  \big(X_i\big)_{i\in I}
  \quad \mapsto \quad
  \Bigg(\sum_{v\in A_j} \ \prod_{u\in D_v} X_{n(u)}\Bigg)_{j\in J}
  \ .
\end{equation}
The only difference is that instead of having a set~$A$ indexing the sum, we
have a~$J$-indexed family of sets, each~$A_j$ indexing the sum of the~$j$
component of the functor. Categorically speaking, it amounts to replacing the
set~$A$ by a function~$\alpha : A \to J$ representing a~$J$-indexed family of
sets.
The following definition now makes sense in any category:
\begin{defi}
  If $I$ and $J$ are objects of~$\C$, a \emph{polynomial diagram} from~$I$
  to~$J$ is a diagram~$P$ in~$\C$ of the shape
  \[
    P :
    \begin{tikzpicture}[baseline=(m-1-1.base)]
      \matrix (m) [matrix of math nodes, column sep=3em,text height=1ex, text
      depth=.25ex]
        { I & D & A & J \\ };
      \path[morphism]
        (m-1-2) edge node[above]{$n$} (m-1-1)
        (m-1-2) edge node[above]{$d$} (m-1-3)
        (m-1-3) edge node[above]{$a$} (m-1-4);
    \end{tikzpicture}
    \ .
  \]
  We write~$\Poly_\C[I,J]$ for the collection of polynomial diagrams from~$I$ to~$J$.
\end{defi}

\begin{defi}
  For each~$P\in\Poly_\C[I,J]$, there is an associated functor~$\Sem P$ from~$\Sl\C
  I$ to~$\Sl\C J$:
  \[
    \Sem P :
    \begin{tikzpicture}[baseline=(m-1-1.base)]
      \matrix (m) [matrix of math nodes, column sep=3em,text height=1ex, text
      depth=.25ex]
        { \Sl\C I & \Sl\C D & \Sl\C A & \Sl\C J \\ };
      \path[morphism]
        (m-1-1) edge node[above]{$\Delta_n$} (m-1-2)
        (m-1-2) edge node[above]{$\Pi_d$} (m-1-3)
        (m-1-3) edge node[above]{$\Sigma_a$} (m-1-4);
    \end{tikzpicture}
    \ .
  \]
  $\Sem P$ is called the \emph{extension} of~$P$ and~$P$ is called a
  \emph{representation} of~$\Sem P$. Any functor naturally isomorphic to
  some~$\Sem P$ is called a \emph{polynomial functor}. We
  write~$\PolyFun_\C[I,J]$ for the collection of such functors.
\end{defi}
In the locally cartesian closed category of sets and functions, the
operations~$\Delta_n$,~$\Pi_d$ and~$\Sigma_a$ are given by:
\begin{itemize}
  \item $\Delta_n : \big(X_i\big)_{i\in I} \mapsto \big(X_{n(d)}\big)_{d\in D}$

  \item $\Pi_d : \big(Y_u\big)_{u\in D} \mapsto \big(\prod_{u\in d^{-1}(a)} Y_u\big)_{a\in A}$

  \item $\Sigma_a : \big(Z_v\big)_{v\in A} \mapsto \big(\sum_{v\in a^{-1}(j)}
    Z_v\big)_{j\in J}$
\end{itemize}
and the extension of a polynomial functor corresponds exactly to the
formula~$(\ref{eqn:ext_in_set})$.
We have
\begin{prop} \label{prop:PFComposition}
  The composition of two polynomial functors is a polynomial functor.
\end{prop}

\begin{proof}[Sketch of proof]
  The composition~$Q\circ P$ of~$P\in\Poly_\C[I,J]$
  and~$Q\in\Poly_\C[J,K]$ uses no less than four pullbacks:
  \begin{equation}\label{diag:composition}
    Q\circ P \quad \eqdef \quad
    \begin{tikzpicture}[baseline=(m-2-3.base)]
      \matrix (m) [matrix of math nodes, column sep={2.5em,between origins},
      row sep={2.5em,between origins},text depth=.25ex, text height=1.5ex]
        {   &   &    &\bullet&    &\cdotp&   &\bullet&   \\
            &   &\cdotp&   &\cdotp&      &   &      &    \\
            & D &      & A &      & E    &   & B    &    \\
          I &   &      &   & J    &      &   &      & K  \\ };
      \draw ($(m-1-6)!.5!(m-3-8)$) node[auto]{$\scriptstyle(i)$};
      \path[morphism]
        (m-1-4) edge (m-1-6)
        (m-1-6) edge (m-1-8)
        (m-1-6) edge (m-3-6)
        (m-1-4) edge (m-2-3)
        (m-1-6) edge node[over]{$\epsilon$} (m-2-5)
        (m-1-8) edge (m-3-8)
        (m-2-3) edge (m-2-5)
        (m-2-3) edge (m-3-2)
        (m-2-5) edge (m-3-4)
        (m-2-5) edge (m-3-6)
        (m-3-2) edge (m-3-4)
        (m-3-2) edge (m-4-1)
        (m-3-4) edge (m-4-5)
        (m-3-6) edge (m-3-8)
        (m-3-6) edge (m-4-5)
        (m-3-8) edge (m-4-9);
      \pullback{m-3-4}{m-2-5}{m-3-6}{draw,-};
      \pullback{m-3-2}{m-2-3}{m-2-5}{draw,-};
      \pullback[.5cm]{m-3-6}{m-1-6}{m-1-8}{draw,-};
      \pullback{m-2-3}{m-1-4}{m-1-6}{draw,-};
    \end{tikzpicture}
  \end{equation}
  where square~$(i)$ is a distributivity square (diagram~$(\ref{diag:distr})$
  in Appendix~\ref{app:LCCC}).
  One can then show that
  \[
    \Sem Q\circ\Sem P \quad \iso \quad \Sem{Q\circ P}
  \]
  by a sequence of Beck-Chevalley and distributivity isomorphisms (see
  Appendix~\ref{app:LCCC}).
\end{proof}

The identity functor from~$\Sl\C I$ to itself is trivially the extension of the
polynomial
\[
  \begin{tikzpicture}[baseline=(m-1-1.base)]
    \matrix (m) [matrix of math nodes, column sep=3em,text height=1ex, text
    depth=.25ex]
      { I & I & I & I \\ };
    \path[morphism]
      (m-1-2) edge node[above]{$1$} (m-1-1)
      (m-1-2) edge node[above]{$1$} (m-1-3)
      (m-1-3) edge node[above]{$1$} (m-1-4);
  \end{tikzpicture}
  \ .
\]
We obtain a \emph{bicategory}~$\Poly_\C$ where objects are objects of~$\C$ and
morphisms are polynomial diagrams; and a \emph{category}~$\PolyFun_\C$ where
objects are slice categories and
morphisms are polynomial functors.


\subsection{Strong Natural Transformations} 

Each polynomial functor~$P$ from~$\Sl\C I$ to~$\Sl\C J$ is equipped
with a \emph{strength}:
\[
  \tau_{A,x} \quad :\quad
  A \odot P(x) \to P(A\odot x)
  \ ,
\]
naturally in~$A\in\C$ and~$x\in\Sl\C I$, where~$A \odot x \eqdef
\Sigma_x\Delta_x\Delta_I(A)$.
A natural transformation between two strong functors is itself \emph{strong}
when it is compatible with their strengths.
This gives the category~$\PolyFun_\C$ a 2-category structure: objects are
slice categories, morphisms are polynomial functors and 2-cells are strong
natural transformations. Note that when the base category~$\C$ is~$\Set$,
all endofunctors are strong, and so are all natural transformations.

The following proposition is crucial as it allows to represent
strong natural transformations between polynomial functors by diagrams inside
the category~$\C$~\cite{polyMonads}.

\begin{prop}\label{prop:strongNat}
  Every strong natural transformation~$\rho : P_1 \Rightarrow P_2$ between polynomial
  functors can be uniquely represented (up-to a choice of pullbacks) by a
  diagram
  \begin{equation}\label{diag:strongNat}
    \begin{tikzpicture}[baseline=(m-2-3).base]
      \matrix (m) [matrix of math nodes, column sep={3.5em,between origins}, row
      sep={3.5em,between origins},text height=1.5ex, text depth=.25ex]
      { P_1 : & I & D_1 & A_1 & J \\
              &   & X   & A_1 &   \\
        P_2 : & I & D_2 & A_2 & J  \\};
      \path[morphism]
        (m-1-3) edge node[above]{$n_1$} (m-1-2)
        (m-1-3) edge node[above]{$d_1$} (m-1-4)
        (m-1-4) edge node[above]{$a_1$} (m-1-5)
        (m-3-3) edge node[below]{$n_2$} (m-3-2)
        (m-3-3) edge node[below]{$d_2$} (m-3-4)
        (m-3-4) edge node[below]{$a_2$} (m-3-5);
      \draw[double equal sign distance,double]
        (m-1-2) -- (m-3-2)
        (m-1-5) -- (m-3-5)
        (m-1-4) -- (m-2-4);
      \path[morphism,densely dashed]
        (m-2-3) edge node[above]{$f$} (m-2-4)
        (m-2-3) edge node[left]{$\beta$} (m-1-3)
        (m-2-3) edge node[left]{$g$} (m-3-3)
        (m-2-4) edge node[right]{$\alpha$} (m-3-4);
      \pullback[.5cm]{m-3-3}{m-2-3}{m-2-4}{draw,-};
    \end{tikzpicture}
    \ .
  \end{equation}
The strong natural transformation associated to such a diagram is defined as
\[\begin{array}{ccll}
    \Sigma_{a_1}\Pi_{d_1}\Delta_{n_1}
  &\implies\Big.&
    \Sigma_{a_1}\Pi_{d_1}\Pi_\beta\Delta_\beta\Delta_{n_1}
    &\qquad\hbox{\scriptsize(unit of~$\Delta_\beta\dashv\Pi_\beta$)}\\
  &\iso\Big.&
    \Sigma_{a_1}\Pi_f\Delta_g\Delta_{n_2}
    &\qquad\hbox{\scriptsize($n_1\beta=n_2 g$ and $d_1 \beta=f$)}\\
  &\iso\Big.&
    \Sigma_{a_1}\Delta_\alpha\Pi_{d_2}\Delta_{n_2}
    &\qquad\hbox{\scriptsize(Beck-Chevalley isomorphism)}\\
  &\iso\Big.&
    \Sigma_{a_2}\Sigma_\alpha\Delta_\alpha\Pi_{d_2}\Delta_{n_2}
    &\qquad\hbox{\scriptsize($a=b\alpha$)}\\
  &\implies\Big.&
    \Sigma_{a_2}\Pi_{d_2}\Delta_{n_2}
    &\qquad\hbox{\scriptsize(counit of~$\Sigma_\alpha\dashv\Delta_\alpha$)}
  \ .
\end{array}\]
\end{prop}
A corollary to Proposition~\ref{prop:strongNat} is
\begin{cor} \label{cor:isoBetweenPolynomials}
  If $\Sem{P_1} \iso \Sem{P_2}$ in~$\PolyFun_\C$, then $P_1$ and $P_2$ are
  related by
  \[
    \begin{tikzpicture}[baseline=(m-2-2).base]
      \matrix (m) [matrix of math nodes, column sep={3.5em,between origins}, row
      sep={3.5em,between origins},,text height=1.5ex, text depth=.25ex]
        { P_1: & I & D_1 & A_1 & J \\
          P_2: & I & D_2 & A_2 & J \\};
      \path[morphism]
        (m-1-3) edge node[above]{$n_1$} (m-1-2)
        (m-1-3) edge node[above]{$d_1$} (m-1-4)
        (m-1-4) edge node[above]{$a_1$} (m-1-5)
        (m-2-3) edge node[below]{$n_2$} (m-2-2)
        (m-2-3) edge node[below]{$d_2$} (m-2-4)
        (m-2-4) edge node[below]{$a_2$} (m-2-5);
    \path[morphism]
        (m-2-3) edge node[above,sloped]{$\thicksim$} (m-1-3)
        (m-1-4) edge node[above,sloped]{$\thicksim$} (m-2-4);
    \draw[double equal sign distance,double]
      (m-1-2) -- (m-2-2)
      (m-1-5) -- (m-2-5);
    \end{tikzpicture}
    \ .
  \]
\end{cor}
This is particularly important as it means that instead of working on
polynomial functors up-to strong natural isomorphism, we can work on
their representing polynomials.


\subsection{Spans and Polynomial Functors} 

There are two ways to lift a span to a polynomial:
\begin{defi}
  Given a span $R = \langle f,g\rangle :I\leftarrow X\to J$, we define two
  polynomials in~$\Poly_\C[I,J]$:
  \[
    \AU{R}\quad\eqdef\quad
    \begin{tikzpicture}[baseline=(m-1-1.base)]
      \matrix (m) [matrix of math nodes, column sep=3em,text height=1ex, text depth=.25ex]
        { I & X & X & J \\ };
      \path[morphism]
        (m-1-2) edge node[above]{$f$} (m-1-1)
        (m-1-2) edge node[above]{$1$} (m-1-3)
        (m-1-3) edge node[above]{$g$} (m-1-4);
    \end{tikzpicture}
    \quad
  \]
  and
  \[
    \DU{R}\quad\eqdef\quad
    \begin{tikzpicture}[baseline=(m-1-1.base)]
      \matrix (m) [matrix of math nodes, column sep=3em,text height=1ex, text depth=.25ex]
        { I & X & J & J \\ };
      \path[morphism]
        (m-1-2) edge node[above]{$f$} (m-1-1)
        (m-1-2) edge node[above]{$g$} (m-1-3)
        (m-1-3) edge node[above]{$1$} (m-1-4);
    \end{tikzpicture}
    \quad .
  \]
  Any functor of the form~$\Sem{\AU{R}}$ is called a \emph{linear} polynomial
  functor.
\end{defi}
The terminology ``linear'' comes from the fact that~$\Sem{\AU{R}}$ commutes
with arbitrary colimits. More precisely, we have
\begin{lem}\label{lem:AULeftAdjointLeft}
  If one writes~$R^\sim$ for the span~$R$ with its ``legs''
  reversed, we have an adjunction
  \[
    \Sem{\,\AU{R}\,} \quad\dashv\quad \Sem{\,\DU{R^\sim}\,}
    \quad.
  \]
\end{lem}

\begin{proof}
  If $R$ is~$\langle f,g\rangle$:
  \[
    \begin{tikzpicture}[baseline=(m-2-3.base)]
      \matrix (m) [matrix of math nodes, column sep=1.5em, row sep=1.5em,text
      height=1ex, text depth=.25ex]
        {   & X &   \\
          I &   & J \\ };
      \path[morphism]
        (m-1-2) edge node[above left]{$f$} (m-2-1)
        (m-1-2) edge node[above right]{$g$} (m-2-3);
    \end{tikzpicture}
    \ ,
  \]
  then the extension of $\AU{R}$ is~$\Delta_f\,\Sigma_g$ and the extension
  of~$\DU{R^\sim}$ is~$\Delta_g\,\Pi_f$. The result follows from the two
  adjunctions~$\Delta_f\dashv \Pi_f$ and~$\Sigma_g\dashv \Delta_g$.
\end{proof}
Composition of spans via pullbacks and composition of polynomials are compatible:
\begin{lem}\label{lem:compositionSpans}
  The operations~$\AU{\BLANK}$ and~$\DU{\BLANK}$ from~$\Span_\C$ to~$\Poly_\C$
  are functorial, in a ``bicategorical'' sense.
\end{lem}

\section{Symmetric Monoidal Closed Structure}  
\label{section:PE}

\subsection{SMCC Structure for Polynomial \emph{Diagrams}} 

We start by recalling the main definition and result from~\cite{polyDiagrams}.
\goodbreak
\begin{defi}
  The category~$\PDSim_\C$ has:
  \begin{itemize}
    \item ``endo'' polynomial diagrams~$I\leftarrow D \rightarrow A \rightarrow I$ as objects
    \item equivalence classes of ``simulation diagrams'' as morphisms, where
     a simulation diagram from~$P_1 = I_1\leftarrow D_1 \rightarrow A_1 \rightarrow
     I_1$ to~$P_2 = I_2\leftarrow D_2 \rightarrow A_2 \rightarrow I_2$ is given by a
     diagram like
  \begin{equation*}\label{diag:sim}
    \begin{tikzpicture}[baseline=(m-2-2)]
      \matrix (m) [matrix of math nodes, column sep={3.5em,between origins}, row
      sep={3.5em,between origins},text height=1.5ex, text depth=.25ex]
      { P_1 : & I_1 & D_1 & A_1 & I_1 \\
              & R &\cdot&\cdot& R \\
        P_2 : &  I_2 & D_2 & A_2 & I_2 & .\\};
      \path[morphism]
        (m-1-3) edge (m-1-2)
        (m-1-3) edge (m-1-4)
        (m-1-4) edge (m-1-5)
        (m-3-3) edge (m-3-2)
        (m-3-3) edge (m-3-4)
        (m-3-4) edge (m-3-5);
      \path[morphism,densely dashed]
        (m-2-2) edge node[left]{$r_1$} (m-1-2)
        (m-2-2) edge node[left]{$r_2$} (m-3-2)
        (m-2-5) edge node[right]{$r_1$} (m-1-5)
        (m-2-5) edge node[right]{$r_2$} (m-3-5)
        (m-2-3) edge node[above]{$\gamma$} (m-2-2)
        (m-2-3) edge (m-2-4)
        (m-2-4) edge (m-2-5)
        (m-2-3) edge node[right]{$\beta$} (m-1-3)
        (m-2-3) edge (m-3-3)
        (m-2-4) edge (m-1-4)
        (m-2-4) edge node[right]{$\alpha$} (m-3-4);
      \pullback[.6cm]{m-1-4}{m-2-4}{m-2-5}{draw,-};
      \pullback[.6cm]{m-3-3}{m-2-3}{m-2-4}{draw,-};
    \end{tikzpicture}
  \end{equation*}
  The equivalence relation between such diagrams is detailed
  in~\cite{polyDiagrams} and corresponds to the equivalence between spans that
  form the sides of simulations.
  \end{itemize}
\end{defi}
\noindent
We have (\cite{polyDiagrams}):
\begin{prop}\label{prop:SMCCDiagrams}
  The operation~$\Tensor$ that acts on objects in a pointwise manner:
\[
  P_1\Tensor P_2 \quad\eqdef\quad
  \begin{tikzpicture}[baseline=(m-1-1.base)]
    \matrix (m) [matrix of math nodes, column sep=3.5em,text height=1ex, text
    depth=.25ex]
      { I_1\times I_2 & D_1\times D_2 & A_1\times A_2 & J_1\times J_2 \\ };
    \path[morphism]
    (m-1-2) edge node[above]{$n_1{\times}n_2$} (m-1-1)
    (m-1-2) edge node[above]{$d_1{\times}d_2$} (m-1-3)
    (m-1-3) edge node[above]{$a_1{\times}a_2$} (m-1-4);
  \end{tikzpicture}
\]
  is a tensor product. It gives the category~$\PDSim_\C$ a symmetric monoidal closed
  structure: there is a functor~$\BLANK \Linear \BLANK :
  \PDSim_\C^\op\times\PDSim_\C\to\PDSim_\C$ and an isomorphism
  \[
    \PDSim_\C[P_1\Tensor P_2\ ,\ P_3]
    \quad\iso\quad
    \PDSim_\C[P_1\ ,\ P_2\Linear P_3]
    \ ,
  \]
  natural in~$P_1$ and~$P_3$.
\end{prop}
Seen from the angle of ``games semantics'' hinted at in the introduction,
this operation is a kind of synchronous, ``lockstep'' parallel composition: a
move in the tensor of two games \emph{must be} a move in each of the games,
and a counter-move / response from the opponent \emph{must be} a response for
each move, in each of the two games.

\subsection{Simulations} 
\label{sub:simulations}

We start by characterizing simulations between polynomials as special 2-cells
involving their extensions.

\begin{prop}\label{prop:sim}
  If $P_1$ and~$P_2$ are polynomial diagrams and~$R$ is a span, any 2-cell
  \[
    \begin{tikzpicture}[baseline=(m-2-1.base)]
      \matrix (m) [matrix of math nodes, column sep=3.5em, row sep=3.5em,text
      height=1ex, text depth=.25ex]
        { \Sl\C{I_1} & \Sl\C{I_1}\\
          \Sl\C{I_2} & \Sl\C{I_2}\\};
      \path[morphism]
        (m-1-1) edge node[above]{$\Sem{P_1}$} (m-1-2)
        (m-2-1) edge node[below]{$\Sem{P_2}$} (m-2-2);
      \path[morphism]
        (m-1-1) edge node[over]{$\Sem{\AU{R}}$} (m-2-1)
        (m-1-2) edge node[over]{$\Sem{\AU{R}}$} (m-2-2);
      \draw[twocell]
        (m-1-2) -- (m-2-1);
      \draw ($(m-1-2)!.5!(m-2-1)$) node[over]{$\scriptstyle\rho$};
    \end{tikzpicture}
  \]
  (where~$\rho$ is a strong natural transformation) is
  uniquely represented (up-to a choice of pullbacks) by a diagram of the shape
  \begin{equation}\label{diag:simbis}
    \begin{tikzpicture}[baseline=(m-2-2)]
      \matrix (m) [matrix of math nodes, column sep={4em,between origins}, row
      sep={4em,between origins},text height=1.5ex, text depth=.25ex]
        {P_1: & I_1 & D_1 & A_1 & I_1 \\
            & R &R{\cdotp}D_2&R{\cdotp}A_1& R \\
         P_2: & I_2 & D_2 & A_2 & I_2 & .\\};
      \path[morphism]
        (m-1-3) edge node[above]{$n_1$} (m-1-2)
        (m-1-3) edge node[above]{$d_1$} (m-1-4)
        (m-1-4) edge node[above]{$a_1$} (m-1-5)
        (m-3-3) edge node[below]{$n_2$} (m-3-2)
        (m-3-3) edge node[below]{$d_2$} (m-3-4)
        (m-3-4) edge node[below]{$a_2$} (m-3-5);
      \path[morphism,densely dashed]
        (m-2-2) edge node[left]{$r_1$} (m-1-2)
        (m-2-2) edge node[left]{$r_2$} (m-3-2)
        (m-2-5) edge node[right]{$r_1$} (m-1-5)
        (m-2-5) edge node[right]{$r_2$} (m-3-5)
        (m-2-3) edge node[above]{$\gamma$} (m-2-2)
        (m-2-3) edge (m-2-4)
        (m-2-4) edge (m-2-5)
        (m-2-3) edge node[right]{$\beta$} (m-1-3)
        (m-2-3) edge (m-3-3)
        (m-2-4) edge (m-1-4)
        (m-2-4) edge node[right]{$\alpha$} (m-3-4);
      \pullback[.6cm]{m-1-4}{m-2-4}{m-2-5}{draw,-};
      \pullback[.6cm]{m-3-3}{m-2-3}{m-2-4}{draw,-};
    \end{tikzpicture}
  \end{equation}
\end{prop}
\begin{proof}
  As compositions of polynomial functors, both~$\Sem{\AU{R}}
  \Sem{P_1}$ and~$\Sem{P_2} \Sem{\AU{R}}$ are polynomial
  (Proposition~~\ref{prop:PFComposition}).
  We can use Proposition~\ref{prop:strongNat} to represent the strong natural
  transformation~$\rho$ by the following diagram:
  \begin{equation}\label{diag:expandedSimulation}
    \begin{tikzpicture}[baseline=(m-3-3.base)]
      \matrix (m) [matrix of math nodes, column sep={4em,between origins}, row
      sep={4em,between origins}, text height=1.5ex,text depth=0.25ex]
        { I_1 &     & D_1  & A_1          & I_1 \\
              &     & X    & R{\cdotp}A_1 & R   \\
          I_1 &     & U    & R{\cdotp}A_1 & I_2 \\
          R   & Y   &\bullet& \bullet       &     \\
          I_2 & D_2 & D_2  & A_2          & I_2 \\};
      \draw ($(m-5-3)!.5!(m-4-4)$) node[auto]{$\scriptstyle(i)$};
      \path[morphism,densely dashed]
        (m-3-3) edge (m-3-4)
        (m-3-3) edge (m-2-3)
        (m-3-3) edge (m-4-3)
        (m-3-4) edge (m-4-4);
      \draw[double equal sign distance,double]
        (m-5-2) -- (m-5-3);
      \draw[double equal sign distance,double]
        (m-3-4) -- (m-2-4);
      \path[morphism]
        (m-1-3) edge node[above]{$n_1$} (m-1-1)
        (m-1-3) edge node[above]{$d_1$} (m-1-4)
        (m-1-4) edge node[above]{$a_1$} (m-1-5)
        (m-5-2) edge node[below]{$n_2$} (m-5-1)
        (m-5-3) edge node[below]{$d_2$} (m-5-4)
        (m-5-4) edge node[below]{$a_2$} (m-5-5)
        (m-4-1) edge node[left]{$r_1$} (m-3-1)
        (m-4-1) edge node[left]{$r_2$} (m-5-1)
        (m-2-5) edge node[right]{$r_1$} (m-1-5)
        (m-2-5) edge node[right]{$r_2$} (m-3-5)
        (m-2-3) edge (m-1-3)
        (m-2-4) edge (m-1-4)
        (m-2-3) edge (m-2-4)
        (m-2-4) edge (m-2-5)
        (m-4-2) edge (m-4-1)
        (m-4-2) edge (m-5-2)
        (m-4-3) edge node[over]{$\epsilon$} (m-4-2)
        (m-4-3) edge (m-5-3)
        (m-4-3) edge (m-4-4)
        (m-4-4) edge (m-5-4) ;
      \draw[double equal sign distance,double]
        (m-1-1) -- (m-3-1)
        (m-3-5) -- (m-5-5);
      \pullback[.4cm]{m-4-1}{m-4-2}{m-5-2}{draw,-};
      \pullback[.4cm]{m-4-4}{m-4-3}{m-5-3}{draw,-};
      \pullback[.4cm]{m-3-4}{m-3-3}{m-4-3}{draw,-};
      \pullback[.4cm]{m-1-3}{m-2-3}{m-2-4}{draw,-};
      \pullback[.4cm]{m-1-4}{m-2-4}{m-2-5}{draw,-};
    \end{tikzpicture}
  \end{equation}
  where~$(i)$ is a distributivity square.
  The plain arrows represent the polynomials~$\AU{R} P_1$ and~$P_2 \AU{R}$,
  whose extensions are~$ \Sem{\AU{R}} \Sem{P_1}$ and~$\Sem{P_2} \Sem{\AU{R}}$; and
  the dashed arrows represent the strong natural transformation between them,
  as in diagram~$(\ref{diag:strongNat})$.

  Going from diagram~$(\ref{diag:expandedSimulation})$ to diagram~$(\ref{diag:simbis})$ is
  easy: just follow the arrows and use the pullback lemma to show that
  square~$(U,R{\cdotp}A_1,A_2,D_2)$ is a pullback and that~$U$ is indeed
  isomorphic to~$R{\cdotp}D_2$. The morphisms~$\alpha$,~$\beta$ and~$\gamma$
  can be read on the diagram.

  Going from diagram~$(\ref{diag:simbis})$ to diagram~$(\ref{diag:expandedSimulation})$
  is slightly messier. We choose~$U$ to be~$R{\cdotp}D_2$ and look at the
  following:
  \[
    \begin{tikzpicture}[baseline=(m-3-3.base)]
      \matrix (m) [matrix of math nodes, column sep={4em,between origins}, row
      sep={4em,between origins}, text height=1.5ex,text depth=0.25ex]
        { I_1 &     & D_1          & A_1          & J_1 \\
              &     & X            & R{\cdotp}A_1 & R   \\
          I_1 &     & R{\cdotp}D_2 & R{\cdotp}A_1 & J_2 \\
          R   & Y   & \cdotp       & \cdotp       &     \\
          I_2 & D_2 & D_2          & A_2          & J_2 \\};
      \draw ($(m-5-3)!.5!(m-4-4)$) node[auto]{$\scriptstyle(i)$};
      \draw ($(m-5-2)!.7!(m-3-3)$) node[auto]{$\scriptstyle(ii)$};
      \path[morphism,densely dashed]
        (m-3-3) edge (m-3-4)
        (m-3-3) edge node[left]{$h$} (m-2-3)
        (m-3-3) edge (m-4-3)
        (m-3-4) edge node[left]{$f$} (m-4-4);
      \draw[double equal sign distance,double]
        (m-5-2) -- (m-5-3);
      \draw[double equal sign distance,double]
        (m-3-4) -- (m-2-4);
      \path[morphism]
        (m-1-3) edge node[above]{$n_1$} (m-1-1)
        (m-1-3) edge node[above]{$d_1$} (m-1-4)
        (m-1-4) edge node[above]{$a_1$} (m-1-5)
        (m-5-2) edge node[below]{$n_2$} (m-5-1)
        (m-5-3) edge node[below]{$d_2$} (m-5-4)
        (m-5-4) edge node[below]{$a_2$} (m-5-5)
        (m-4-1) edge node[left]{$l_1$} (m-3-1)
        (m-4-1) edge node[left]{$l_2$} (m-5-1)
        (m-2-5) edge node[right]{$r_1$} (m-1-5)
        (m-2-5) edge node[right]{$r_2$} (m-3-5)
        (m-2-3) edge (m-1-3)
        (m-2-4) edge (m-1-4)
        (m-2-3) edge (m-2-4)
        (m-2-4) edge (m-2-5)
        (m-4-2) edge (m-4-1)
        (m-4-2) edge node[right]{$d$} (m-5-2)
        (m-4-3) edge node[over]{$\epsilon$} (m-4-2)
        (m-4-3) edge (m-5-3)
        (m-4-3) edge (m-4-4)
        (m-4-4) edge (m-5-4) ;
      \draw[double equal sign distance,double]
        (m-1-1) -- (m-3-1)
        (m-3-5) -- (m-5-5);
      \path[morphism]
        (m-3-3) edge[bend left=-25] node[above left]{$\gamma$} (m-4-1)
        (m-3-3) edge[bend left=35] node[below left]{$\beta$} (m-1-3)
        (m-3-4) edge[bend left=25] node[right]{$\alpha$} (m-5-4);
      \path[morphism,densely dashed]
        (m-3-3) edge[bend left=-20] node[above left]{$g$} (m-4-2);
      \pullback[.5cm]{m-4-1}{m-4-2}{m-5-2}{draw,-};
      \pullback[.5cm]{m-4-4}{m-4-3}{m-5-3}{draw,-};
      \pullback[.5cm]{m-3-4}{m-3-3}{m-4-3}{draw,-};
      \pullback[.5cm]{m-1-3}{m-2-3}{m-2-4}{draw,-};
      \pullback[.5cm]{m-1-4}{m-2-4}{m-2-5}{draw,-};
    \end{tikzpicture}
  \]
  The morphisms~$\alpha$,~$\beta$ and~$\gamma$ come from
  diagram~$(\ref{diag:simbis})$.

  To define~$f$ and~$h$, write~$\varphi$ for the natural isomorphism~$\Sl\C
  U[\Sigma_k\BLANK,\BLANK] \implies \Sl\C V[\BLANK,\Delta_k\BLANK]$ and~$\psi$
  for the natural isomorphism~$\Sl\C V[\Delta_k\BLANK,\BLANK]\implies\Sl\C
  U[\BLANK,\Pi_k\BLANK]$. In particular,~$\epsilon$ is~$\psi^{-1}(1)$. Then:
  \begin{itemize}

  \item $f$ is constructed from~$\alpha$ and~$\gamma$:
  \[\begin{array}{cll}
         & \Big.
     \gamma \ \in\  \Sl\C{I_2}\big[\ \Sigma_{n_2}\Delta_{d_2}(\alpha)  \ ,\ l_2\ \big]
         & \quad\hbox{\footnotesize see diagram~$(\ref{diag:simbis})$} \\
    \iff & \Big.
     g \eqdef \varphi(\gamma) \ \in\  \Sl\C{D_2}\big[\ \Delta_{d_2}(\alpha)  \ ,\ \Delta_{n_2}(l_2)\ \big]
         & \\
    \iff & \Big.
     f \eqdef \psi \varphi(\gamma) \ \in\  \Sl\C{A_2}\big[\ \alpha  \ ,\ \Pi_{d_2}\Delta_{n_2}(l_2)\ \big]
     \ .
  \end{array}\]
  By naturality of~$\psi^{-1}$, the following commutes:
  \[
    \begin{tikzpicture}[baseline=(m-2-2.base)]
      \matrix (m) [matrix of math nodes, column sep=3em, row
      sep=3em, text height=1.5ex,text depth=0.25ex]
        {
            \Sl\C{A_2}\big[\Pi_{d_2}(d) , \Pi_{d_2}(d)\big]
          & \Sl\C{D_2}\big[\Delta_{d_2}\Pi_{d_2}(d),d\big] \\
            \Sl\C{A_2}\big[\alpha,\Pi_{d_2}(d)\big]
          & \Sl\C{D_2}\big[\Delta_{d_2}(\alpha),d\big] \\
        };
      \path[morphism]
        (m-1-1) edge node[above]{$\psi^{-1}$} (m-1-2)
        (m-1-1) edge node[left]{$\BLANK\circ f$} (m-2-1)
        (m-1-2) edge node[right]{$\BLANK\circ \overrightarrow{\Delta}_{d_2}(f)$} (m-2-2)
        (m-2-1) edge node[below]{$\psi^{-1}$} (m-2-2);
    \end{tikzpicture}
  \]
  where~$\overrightarrow{\Delta}_{d_2}(\BLANK)$ is the action of the
  functor~$\Delta_{d_2}$ on morphisms.  Starting from the identity,
  we get~$\psi^{-1}(1)\circ \overrightarrow{\Delta}_{d_2}(f) =
  \psi^{-1}(f)$, i.e., $\epsilon\circ\overrightarrow{\Delta}_{d_2}(f) = g$. This
  shows that the triangle~$(ii)$ commutes.

  \item We can then construct~$h$ from~$\beta$ by using the fact that~$X$ is a
  pullback.
  \end{itemize}
  The only remaining thing to check is that the upper left rectangle commutes.
  It follows from~$r_1\gamma = n_1\beta$ in diagram~$(\ref{diag:simbis})$ and
  the construction.
\end{proof}
Proposition~\ref{prop:sim} makes it natural to define simulations for
arbitrary endofunctors:
\begin{defi}\label{def:simulationAsNatTransformation}
  If $F_1$ and $F_2$ are two endofunctors over~$\Sl\C{I_1}$ and~$\Sl\C{I_2}$,
  a \emph{simulation from~$F_1$ to~$F_2$} is given by a span~$I_1\leftarrow
  R\rightarrow I_2$ and a 2-cell of the form
  \[
    \begin{tikzpicture}[baseline=(m-2-1.base)]
      \matrix (m) [matrix of math nodes, column sep=3.5em, row sep=3.5em,text
      height=1ex, text depth=.25ex]
        { \Sl\C{I_1} & \Sl\C{I_1}\\
          \Sl\C{I_2} & \Sl\C{I_2}\\};
      \path[morphism]
        (m-1-1) edge node[above]{$F_1$} (m-1-2)
        (m-2-1) edge node[below]{$F_2$} (m-2-2);
      \path[morphism]
        (m-1-1) edge node[over]{$\Sem{\AU{R}}$} (m-2-1)
        (m-1-2) edge node[over]{$\Sem{\AU{R}}$} (m-2-2);
      \draw[twocell]
        (m-1-2) -- (m-2-1);
      \draw ($(m-1-2)!.5!(m-2-1)$) node[over]{$\scriptstyle\rho$};
    \end{tikzpicture}
  \]
  where~$\rho$ is a natural transformation.
  A simulation~$(R,\rho)$ is \emph{equivalent} to a simulation~$(R',\rho')$
  iff~$F_2\varepsilon \circ \rho = \rho' \circ \varepsilon F_1$ for some
  natural isomorphism~$\varepsilon : \Sem{\AU{R}} \to \Sem{\AU{R'}}$.

  The category~$\FSim_\C$ is defined with:
    \begin{itemize}
      \item endofunctors over slices of~$\C$ as objects,
      \item equivalence classes of simulations as morphisms.
    \end{itemize}
\end{defi}
It is routine to check that this gives a category. (Recall that the
composition of two linear functors is again linear by
Lemma~\ref{lem:compositionSpans}).
Note that the notion of equivalence of simulations is inherited from~$\Span_\C$: by
Corollary~\ref{cor:isoBetweenPolynomials} an isomorphism
between~$\Sem{\AU{R}}$ and~$\Sem{\AU{R'}}$ amounts to a span isomorphism
between~$R$ and~$R'$.
As a particular subcategory, we have
\begin{defi}
  The category~$\PFSim_\C$ is the subcategory of~$\FSim_\C$ with
    \begin{itemize}
      \item \emph{polynomial} endofunctors over slices of~$\C$ as objects,
      \item equivalence classes of \emph{strong} simulations as morphisms
    \end{itemize}
  where a simulation~$(R,\rho)$ is strong if and only if the natural
  transformation~$\rho$ is strong.
\end{defi}
There is a functor from~$\PDSim_\C$ to~$\PFSim_\C$ that sends a polynomial
diagram to its corresponding polynomial functor and a simulation diagram to
its simulation cell. This functor is
\begin{itemize}
  \item surjective on objects by the definition of polynomial functor,
  \item full and faithful by Proposition~\ref{prop:sim}.
\end{itemize}
This implies that
\begin{prop}\label{prop:PEequivPEFun}
  The category~$\PDSim_\C$ and the category~$\PFSim_\C$ are equivalent.
\end{prop}
%


\subsection{Tensor Product} 

Recall that the tensor of two polynomials is the ``pointwise
cartesian product'':
\[
  P_1\Tensor P_2 \quad\eqdef\quad
  \begin{tikzpicture}[baseline=(m-1-1.base)]
    \matrix (m) [matrix of math nodes, column sep=3.5em,text height=1ex, text
    depth=.25ex]
      { I_1\times I_2 & D_1\times D_2 & A_1\times A_2 & I_1\times I_2 \\ };
    \path[morphism]
    (m-1-2) edge node[above]{$n_1{\times}n_2$} (m-1-1)
    (m-1-2) edge node[above]{$d_1{\times}d_2$} (m-1-3)
    (m-1-3) edge node[above]{$a_1{\times}a_2$} (m-1-4);
  \end{tikzpicture}
  \ .
\]
This gives rise to an operation on polynomial
functors:~$\Sem{P_1}\Tensor\Sem{P_2} \eqdef \Sem{P_1 \Tensor P_2}$. However,
this definition is intensional because it acts on polynomial diagrams, i.e.,
on representations of polynomial functors. In particular, it doesn't even make
sense for functors that are not polynomial.
We will now show that it is possible to characterize~$\Sem{P_1\Tensor P_2}$ by a
universal property relying only on~$\Sem{P_1}$ and~$\Sem{P_2}$, thus
giving an extensional definition of the tensor of polynomial functors.

To avoid confusion, we will write~$\sq : \Sl\C A\times \Sl\C B \to
\Sl\C{A{\times}B}$ for the functor sending~$(x,y)$ in~$\Sl\C A\times\Sl\C B$
to~$x\times y$ in~$\Sl\C{A{\times}B}$ and~$(f,g)$ in~$\in\Sl\C
A[x,x']\times\Sl\C B[y,y']$ to~$f\times g$
in~$\Sl\C{A{\times}B}[x{\times}x',y{\times}y']$, we have:
\begin{prop}\label{prop:tensorLan}
  Let $P_1$ and $P_2$ be polynomial functors,
  the polynomial functor~$P_1\Tensor P_2$ is a left Kan-extension along~$\sq$:
  it is universal s.t.
  \[
    \begin{tikzpicture}[baseline=(m-2-2.base)]
      \matrix (m) [matrix of math nodes, column sep=3.5em, row sep=3.5em,text
      height=1ex, text depth=.25ex]
        {
          \Sl\C{I_1}\times\Sl\C{I_2}  &  \Sl\C{I_1{\times}I_2} \\
          \Sl\C{I_1}\times\Sl\C{I_2}  &  \Sl\C{I_1{\times}I_2} \\
        };
      \path[morphism]
        (m-1-1) edge node[above]{$\sq$} (m-1-2)
        (m-2-1) edge node[below]{$\sq$} (m-2-2)
        (m-1-1) edge node[left]{$P_1 \times P_2$} (m-2-1)
        (m-1-2) edge node[right]{$P_1 \Tensor P_2$} (m-2-2);
      \draw[twocell]
        (m-2-1) -- node[above left]{$\varepsilon$} (m-1-2);
    \end{tikzpicture}
    \ .
  \]
  More precisely, $P_1 \Tensor P_2 = \Lan_{\sq} \big(P_1(\BLANK) \sq
  P_2(\BLANK)\big)$ in the category of endofunctors with natural
  transformations.
\end{prop}
\begin{cor}\label{cor:TensorCoend}
  If~$\C$ has copowers, denoted by~$\odot$, we can express left
  Kan-extensions using coends. We then have
  \[
    P_1{\Tensor}P_2(r)
    \quad = \quad
    \Coend{x,y} \Sl\C{I_1{\times}I_2}[x\sq y,r] \odot \big(P_1(x)\sq P_2(y)\big)
    \ .
  \]
\end{cor}
This definition is reminiscent of the tensor of predicate transformers
(Definition~7 in~\cite{denotPT}):
if~$P_1:\Pow(S_1) \to \Pow(S_1)$ and $P_2:\Pow(S_2) \to \Pow(S_2)$ are
monotonic operators on subsets, then
{\everymath{\displaystyle}
\[\begin{array}{rccccl}
  P_1{\Tensor}P_2 &:& \Pow(S_1\times S_2) &\to& \Pow(S_1\times S_2)\\
                  & & r                 &\mapsto&
        \bigcup_{x\times y\subseteq r} P_1(x)\times P_2(y) &
   .
\end{array}\]}

The proof of proposition~\ref{prop:tensorLan} makes heavy use of the
internal language of locally cartesian closed categories.
%
First note that a polynomial~$I\leftarrow D\to A\to I$ can
be described by the following judgments:
\begin{enumerate}\label{internalPoly}
  \item ``$\vdash I\ \mathsf{type}$'' for the object~$I$ (slice over~$\One$),
  \item ``$i:I \vdash A(i)\ \mathsf{type}$'' for the slice~$a: A\to I$ in~$\Sl\C I$,
  \item ``$i:I, a:A(i) \vdash D(i,a)\ \mathsf{type}$'' for the slice~$d: D\to A$ in~$\Sl\C
  A$,
  \item ``$i:I, a:A(i), d:D(i,a) \vdash n(i,a,d) : I$'' for the morphism~$n:D\to I$.
\end{enumerate}
\begin{proof}[Proof of proposition~\ref{prop:tensorLan}]

  Because~$\Sigma_{a_1\times a_2}(\BLANK\sq\BLANK) =
  \Sigma_{a_1}(\BLANK)\sq\Sigma_{a_2}(\BLANK)$ is a left-adjoint, it
  commutes with all colimits, including left Kan-extensions. We thus have
  \begin{eqnarray*}
    \Lan_{\sq} \big(P_1(\BLANK) \sq P_2(\BLANK)\big)
    &=&
    \Lan_{\sq} \Big(\Sigma_{a_1}\Pi_{d_1}\Delta_{n_1}(\BLANK) \sq
    \Sigma_{a_2}\Pi_{d_2}\Delta_{n_2}(\BLANK)\Big)\\
    &=&
    \Lan_{\sq} \Big(\big(\Sigma_{a_1}\sq\Sigma_{a_2}\big)
    \big(\Pi_{d_1}\Delta_{n_1}(\BLANK) \sq
    \Pi_{d_2}\Delta_{n_2}(\BLANK)\big)\Big)\\
    &=&
    \Lan_{\sq} \Big(\Sigma_{a_1{\times}a_2}
    \big(\Pi_{d_1}\Delta_{n_1}(\BLANK) \sq
    \Pi_{d_2}\Delta_{n_2}(\BLANK)\big)\Big)\\
    &=&
    \Sigma_{a_1{\times}a_2}\Lan_{\sq} \Big(
    \Pi_{d_1}\Delta_{n_1}(\BLANK) \sq
    \Pi_{d_2}\Delta_{n_2}(\BLANK)\Big)
    \ .
  \end{eqnarray*}
  To save some parenthesis, we will write~$n_1\cdotp d_1$
  instead of~$n_1(i_1,a_1,d_1)$ and similarly for~$n_2\cdotp d_2$.
  We write~$F_1$ and~$F_2$ for~$\Pi_{d_1}\Delta_{n_1}$
  and~$\Pi_{d_2}\Delta_{n_2}$. Internally,~$F_1$ is
  thus~$X\mapsto\prod_{d_1}X(n_1\cdotp d_1)$ and~$F_1\Tensor F_2$
  is~$R\mapsto \prod_{d_1,d_2} R(n_1\cdotp d_1,n_2\cdotp d_2)$.

  To reduce verbosity, we will ignore the dependency on~$A_1$ and~$A_2$, i.e.,
  we'll ``pretend'' both are equal to~$\One$. To correct that, one simply
  needs to add~``$a_1{:}A_1(i_1), a_2{:}A_2(i_2)$'' to all contexts and make
  the constructions depend on~$a_1$ and~$a_2$.  Recall that if~$X$ and~$V$ are
  families indexed by~$U$ and~$V$,~$X\sq Y$ is the
  family~``$\fun\langle u,v\rangle . X(u)\times Y(v)$'' indexed by~$U\times V$. We
  use the same notation for functions:~$f\sq g$ stands for~$\fun\langle
  u,v\rangle . \langle f(u),g(v)\rangle$.
  We will
  \begin{enumerate}

    \item construct a natural transformation~$\varepsilon : F_1(\BLANK)\sq F_2(\BLANK)
    \implies F_1{\Tensor}F_2(\BLANK\sq\BLANK)$,

    \item show that~$\varepsilon$ is universal: if~$\rho:
      F_1(\BLANK)\sq F_2(\BLANK)
      \implies
      F(\BLANK\sq\BLANK)$,
      we construct a
      unique transformation~$\Theta:F_1{\Tensor}F_2(\BLANK)\implies
      F(\BLANK)$ such that~$\Theta\varepsilon=\rho$.
  \end{enumerate}
  In the internal language, if~$F$ and~$G$ are two functors from~$\Sl\C U$
  to~$\Sl\C V$, a natural transformation~$\alpha$ from~$F$ to~$G$, takes the
  form of
  \begin{itemize}
    \item families of functions~$v:V\vdash\alpha_{X}(v) : F(X)(v)\to G(X)(v)$
    for any~$U$-indexed type~$X$,
    \item subject to naturality: if~$f:X\to Y$, then~$\alpha_Y F_f=G_f
    \alpha_X$, i.e.,
    \[
      \alpha_Y(v)\big(F_f(v)(y)\big)
      \quad = \quad
      G_f(v)\big(\alpha_X(v)(y)\big)
    \]
    whenever~$u:U\vdash X(u)$ and~$\vdash v:V$.
  \end{itemize}

  The transformation~$\varepsilon:
  \Pi_{d_1}\Delta_{n_1}(\BLANK)\sq \Pi_{d_2}\Delta_{n_2}(\BLANK)
  \implies
  \Pi_{d_1{\times}d_2}\Delta_{n_1{\times}n_2}(\BLANK{\sq}\BLANK)$
  is defined as follows: for
  families~$X$ and~$Y$ over~$I_1$ and~$I_2$
  and~$\langle h_1,h_2\rangle:\prod_{d_1}X\big(n_1\cdotp
  d_1\big)\times \prod_{d_2}Y\big(n_2\cdotp d_2\big)$,
  we put
  \[ \vdash
    \varepsilon_{X,Y}\big\langle h_1,h_2\big\rangle
    \eqdef
    \fun \langle d_1,d_2\rangle . \big\langle h_1(d_1),h_2(d_2)\big\rangle
    : \prod_{d_1,d_2}X\big(n_1\cdotp d_1\big)\times Y\big(n_2\cdotp d_2\big)
    \ .
  \]
  It is easy to check that this is natural. (Categorically
  speaking,~$\varepsilon_{X,Y}\langle h_1,d_2\rangle$ is just~$h_1\times h_2$.)

  To check universality, let~$\rho:
      \Pi_{d_1}\Delta_{n_1}(\BLANK)\sq \Pi_{d_2}\Delta_{n_2}(\BLANK)
      \implies
      F(\BLANK{\sq}\BLANK)$
  for some functor~$F$. We define the natural transformation~$\Theta :
  \Pi_{d_1{\times}d_2}\Delta_{n_1{\times}n_2}(\BLANK)\implies
  F(\BLANK)$ in several steps:
  \begin{enumerate}

    \item define the type~$E_1$ indexed by~$I_1$ as~``$i:I_1 \vdash E_1(i) \eqdef
      \sum_{d_1:D_1} \Id{\big.n_1\cdotp d_1}{i}$'' and similarly for~``$i:I_2\vdash
      E_2(i)$''.\footnote{$\Id{\BLANK}{\BLANK}$ is the extensional
      identity type. Its introduction rule is~``$a:A\vdash
      \refl(a):\Id{a}{a}$''.}

    \item Thus,~$\vdash \rho_{E_1,E_2} :
    \prod_{d_1}E_1\big(n_1\cdotp d_1\big) \sq \prod_{d_2}E_2\big(n_2\cdotp d_2\big)
    \to F(E_1 \sq E_2)$.

    \item Moreover, we have
      \[
        \vdash f_1 \eqdef \fun d_1 . \big\langle d_1, \refl\big(n_1\cdotp d_1\big)\big\rangle : \prod_{d_1} E_1\big(n_1\cdotp d_1\big)
      \]
    and similarly for~``$\vdash f_2$''.

    \item Given~$h$ of type~$\prod_{d_1,d_2} R\big(n_1\cdotp d_1,n_2\cdotp d_2\big)$, we
    construct
    \[
      h,i_1,i_2 \quad  \vdash \quad \overline{h}(i_1,i_2)
      \eqdef
      \fun
       \big\langle
         \langle d_1, e_1 \rangle ,
         \langle d_2, e_2 \rangle
       \big\rangle
       \ .\ %
       h \langle d_1,d_2 \rangle
    \]
    of type~$\prod_{i_1,i_2} E_1(i_1)\times E_2(i_2)\to R(i_1,i_2)$, i.e., of
    type~$E_1\sq E_2 \to R$. It works because~$h(d_1,d_2)$ is of
    type~$R\big(n_1\cdotp d_1,n_2\cdotp d_2\big)$ and since $e_k : \Id{n_k\cdotp
    d_k}{i_k}$ ($k=1,2$), we can
    substitute~$n_1\cdotp d_1$ and~$n_2\cdotp d_2$ for~$i_1$
    and~$i_2$.\footnote{Strictly speaking, we need to compose with an
    isomorphism as substitution works only ``up-to canonical isomorphisms'',
    see~\cite{Hofmann}.} 

    \item We can now define~$\Theta_R$:
      \[
        h \quad\vdash\quad
        \Theta_{R}\big( h \big)
        \ \eqdef\ %
        F_{\overline{h}}\Big( \rho_{E_1,E_2} \big\langle f_1, f_2\big\rangle  \Big)
        \ .
      \]
      This is well typed because~$\rho_{E_1,E_2} \big\langle f_1,
      f_2\big\rangle$ is of type~$F(E_1\sq E_2)$ by points~2,~3
      and because~$F_{\overline{h}}$ is of type~$F(E_1\sq E_2)\to F(R)$.

  \end{enumerate}
  We have~$\Theta\varepsilon = \rho$ because:
  \begin{eqnarray*}
     \Theta_{X\times Y}\Big(\varepsilon_{X,Y}\big\langle h_1,h_2\big\rangle\Big)
       & \quad = \quad &
      F_{\overline{h_1{\times}h_2}}\Big(\rho_{E_1,E_2}\big\langle f_1,f_2\big\rangle\Big)\\
       & \quad = \quad &
    \rho_{X,Y}\big\langle h_1 , h_2 \big\rangle
  \end{eqnarray*}
  where the first equality is the definition of~$\Theta$ and the second
  follows from naturality of~$\rho$:
  \[
    \begin{tikzpicture}[baseline=(m-2-1.base)]
      \matrix (m) [matrix of math nodes, column sep=6.5em, row sep=3.5em,text height=1.5ex, text depth=.25ex]
        {  F_1(E_1)\sq F_2(E_2) & F_1(X)\sq F_2(Y) \\
           F(E_1\sq E_2)        & F(X\sq Y) \\};
      \path[morphism]
        (m-1-1) edge node[above]
        {$F_{1\,\overline{h_1}}\sq F_{2\,\overline{h_2}}$} (m-1-2)
        (m-1-1) edge node[left]  {$\rho_{E_1,E_2}$} (m-2-1)
        (m-1-2) edge node[right]{$\rho_{X,Y}$} (m-2-2)
        (m-2-1) edge node[below]{$F_{\overline{h_1{\times}h_2}}$} (m-2-2);
    \end{tikzpicture}
    \ .
  \]
  It works because the action of~$F_1$ on morphisms is
  composition:~$F_{1\,f}(h) = f\circ h$, where~$f$ is of type~$\prod_i X(i)\to Y(i)$
  and~$h : \prod_{d_1}X\big(n_1\cdotp d_1\big)$. With that in mind, we find that
  \begin{eqnarray*}
    F_{1\,\overline{h_1}}\sq F_{2\,\overline{h_2}} \langle f_1,f_2\rangle
      &\quad=\quad&
    \big\langle\, %
      \overline{h_1}\circ f_1 \ ,\  \overline{h_2}\circ f_2
    \,\big\rangle\\
      &\quad=\quad&
    \big\langle\, %
      \fun d_1 . \overline{h_1}\langle d_1,\refl(n_1\cdotp d_1) \rangle
      \ ,\ %
      \dots
    \,\big\rangle\\
      &\quad=\quad&
    \big\langle\, %
      \fun d_1 . h_1(d_1) \ ,\  \dots
    \,\big\rangle\\
      &\quad=\quad&
    \big\langle\, h_1 \ ,\  h_2 \,\big\rangle
    \ .
  \end{eqnarray*}
  We now need to show that~$\Theta$ is unique with this property. It follows
  from the fact that~$\Theta$ is determined by its values on
  ``rectangles''~$X\sq Y$:
  \begin{eqnarray*}
    \Theta_R(h)
    &=&
    F_{\overline{h}}\Big(\Theta_{E_1\sq E_2} \big(f_1 \times f_2\big) \Big) \\
    &=&
    F_{\overline{h}}\Big(\rho_{E_1\sq E_2}\langle f_1,f_2\rangle\Big)\ .
  \end{eqnarray*}
  The second equality comes from~$\Theta\varepsilon=\rho$ and
  the first one follows from the naturality square
  \[
    \begin{tikzpicture}[baseline=(m-2-1.base)]
      \matrix (m) [matrix of math nodes, column sep=6.5em, row sep=3.5em,text height=1.5ex, text depth=.25ex]
        {  F_1{\Tensor}F_2(E_1{\sq}E_2)     & F_1{\Tensor}F_2(R) \\
           F(E_1{\sq}E_2)                   & F(R) \\ };
      \path[morphism]
        (m-1-1) edge node[above]{$(F_1{\Tensor}F_2)_{\overline{h}}$} (m-1-2)
        (m-1-1) edge node[left]  {$\Theta_{E_1{\sq}E_2}$} (m-2-1)
        (m-1-2) edge node[right]{$\Theta_{R}$} (m-2-2)
        (m-2-1) edge node[below]{$F_{\overline{h}}$} (m-2-2);
    \end{tikzpicture}
  \]
  where like above, we have~$(F_1{\Tensor}F_2)_{\overline{h}}(f_1\sq f_2) = h$.

  This concludes the proof that~$F_1\Tensor F_2=\Lan_{\sq}\big(F_1(\BLANK)\sq
  F_2(\BLANK)\big)$ and thus the proof that~$P_1\Tensor P_2 =
  \Lan_{\sq}\big(P_1(\BLANK)\sq P_2(\BLANK)\big)$.
\end{proof}

This operation is a tensor product. This follows for example from the fact
that it is functorial in~$\PDSim_\C$ and that~$\PFSim_\C$ is equivalent to it
(Proposition~\ref{prop:PEequivPEFun}), but a direct proof is also possible.


\subsection{SMCC Structure} 

From Propositions~\ref{prop:SMCCDiagrams}, \ref{prop:PEequivPEFun}
and~\ref{prop:tensorLan}, we can deduce that
\begin{prop}\label{prop:linear}
  The category~$\PFSim_\C$ with~$\Tensor$ is symmetric
  monoidal closed, i.e., there is a functor~$\BLANK\Linear\BLANK$
  from~$\PFSim_\C^{\mathrm{op}}\times\PFSim_\C$ to~$\PFSim_\C$ with an adjunction
  \[
    \PFSim_\C[P_1\Tensor P_2\ ,\ P_3]
    \quad\iso\quad
    \PFSim_\C[P_1\ ,\ P_2\Linear P_3]
    \ ,
  \]
  natural in~$P_1$ and~$P_3$.
\end{prop}
The concrete intensional definition of~$P_2\Linear P_3$, either in its type
theory version or its diagramatic version is rather verbose (Definition~3.7 or
Lemma~3.8 in~\cite{polyDiagrams}) and won't be needed here. However just as
with the tensor, it is possible to define~$P_2\Linear P_3$ without referring
to the representing polynomial diagrams. Not surprisingly, it takes the form a
right Kan-extension. To simplify the proof, we only state the result for the
case~$\C=\Set$:
\begin{prop}\label{prop:linearRan}
  Given two polynomial endofunctors~$P_2$ and~$P_3$ respectively on~$\Sl\Set{I_2}$
  and~$\Sl\Set{I_3}$,
  the polynomial endofunctor~$P_2\Linear P_3$ on~$\Sl\Set{I_2\times I_3}$ is a right Kan-extension
  along~$\tr$: it is universal such that
  \[
    \begin{tikzpicture}[baseline=(m-2-2.base)]
      \matrix (m) [matrix of math nodes, column sep=3.5em, row sep=3.5em,text
      height=1ex, text depth=.25ex]
        {
          \Sl\Set{I_2}\times\Sl\Set{I_3}  &  \Sl\Set{I_2{\times}I_3} \\
          \Sl\Set{I_2}\times\Sl\Set{I_3}  &  \Sl\Set{I_2{\times}I_3} \\
        };
      \path[morphism]
        (m-1-1) edge node[above]{$\tr$} (m-1-2)
        (m-2-1) edge node[below]{$\tr$} (m-2-2)
        (m-1-1) edge node[left]{$P_2 \times P_3$} (m-2-1)
        (m-1-2) edge node[right]{$P_2 \Linear P_3$} (m-2-2);
      \draw[twocell]
        (m-1-2) -- node[above left]{$\eta$} (m-2-1);
    \end{tikzpicture}
  \]
  where~$\tr$ is defined as~$f\tr g \eqdef \Pi_{f{\times}1}(1{\times}g)$, or
  equivalently,~$f\tr g \eqdef \Pi_{f{\times}1}\Delta_{\pi_2}(g)$.
  More precisely, we have $P_2 \Linear P_3 = \Ran_{\tr}
  \big(P_2(\BLANK)\tr P_3(\BLANK)\big)$ in the category of endofunctors with natural
  transformations.
\end{prop}

\begin{cor}\label{cor:LinearEnd}
  We have
  \[
    P_2\Linear P_3(r)
    \quad = \quad
    \End{y,z} \Sl\Set{I_2{\times}I_3}[r,y\tr z] \pitchfork \big( P_2(y) \tr
    P_3(z) \big)
  \]
  where~$\pitchfork$ is the ``power'' operation.
\end{cor}
Note that in the internal language,~$Y\tr Z$ is~``$i_2:I_2,i_3:I_3 \vdash
Y(i_2)\to Z(i_3)$''. Before proving Proposition~\ref{prop:linearRan}, we show:
\begin{lem}\label{lem:AULeftAdjointRight}
  There is a natural isomorphism
  \[
    \Sl\C{I_3}\big[\AU{r}(y)\,,\,z\big]
    \quad\iso\quad
    \Sl\C{I_2{\times}I_3}\big[r\,,\,y\tr z\big]
    \ .
  \]
\end{lem}
\begin{proof}

  In the internal language, those homsets correspond to the types
  \begin{itemize}
    \item $\prod_{i_3}\big(\sum_{i_2} R(i_2,i_3)\times Y(i_2)\big) \to Z(i_3)$
    \item and~$\prod_{i_2,i_3} R(i_2,i_3)\to\big(Y(i_2)\to Z(i_3)\big)$,
  \end{itemize}
which are indeed isomorphic.
\end{proof}

\begin{proof}[Proof of Proposition~\ref{prop:linearRan}]
  We will show, using the formulas for Kan extensions and the calculus of ends
  and coends \cite{MacLane}, that the adjoint to~$P_1\Tensor\BLANK$ (as
  given by proposition~\ref{prop:tensorLan}) is necessarily the above right
  Kan-extension.

  Suppose~$P_1$, $P_2$ and~$P_3$ are polynomial functors with
  domains~$I_1$,~$I_2$ and~$I_3$. Let~$R$ be a span between~$I_1\times I_2$
  an~$I_3$. We write~$R'$ for the corresponding span between~$I_1$
  and~$I_2\times I_3$. Besides the previous lemmas and propositions, we will
  use:
  \begin{itemize}

    \item if~$K_1\dashv K_2$, there is a natural
    isomorphism~$\Nat(FK_2,G)\iso\Nat(F,GK_1)$ for all functors~$F$ and~$G$ (note
    the inversion of left and right);

    \item the power~$X\pitchfork\BLANK$ is right-adjoint to the
    copower~$X\odot\BLANK$;

    \item if~$K_1\dashv K_2$, then~$[X\odot A , K_2(B)] \iso [X\odot K_1(A),B]$, and
    similarly for~$\pitchfork$;

    \item $\AU{\big.\AU{R'(x)}}(y) \iso \AU{R}(x\sq y)$: in the
    internal language, they are respectively
    \[
      i_3 \quad \vdash \quad
      \sum_{i_2} \left(\sum_{i_1} R(i_1,i_2,i_3)\times X(i_1)\right) \times Y(i_2)
    \]
    and
    \[
      i_3 \quad \vdash \quad
      \sum_{i_1,i_2} R(i_1,i_2,i_3)\times \big(X(i_1)\times Y(i_2)\big)
    \]
    which are naturally isomorphic.

  \end{itemize}

  We have:
{\allowdisplaybreaks
  \begin{eqnarray*}
    & &
    \Nat\Big(\AU{R'} P_1 \ ,\ \Ran_{\tr}\big(P_2(\BLANK)\tr P_3(\BLANK)\big)\AU{R'}\Big)\\
    &\iso&
    \End{x} \Sl{\Set}{I_2\times I_3}\Big[\AU{R'} P_1(x) \ ,\ \Ran_{\tr}\big(P_2(\BLANK)\tr P_3(\BLANK)\big)\AU{R'}(x)\Big]\\
    &\iso&
    \End{x} \Sl{\Set}{I_2\times I_3}\Big[\AU{R'} P_1(x) \ ,\ \End{y,z}
      \Sl{\Set}{I_2\times I_3}\big[\AU{R'}(x),y\tr z\big] \pitchfork
    \big(P_2(y)\tr P_3(z)\big)\  \Big]\\
    &\iso&
    \End{x,y,z} \Sl{\Set}{I_2\times I_3}\Big[\AU{R'} P_1(x) \ ,\ \Sl{\Set}{I_2\times I_3}\big[\AU{R'}(x),y\tr z\big] \pitchfork \big(P_2(y)\tr P_3(z)\big)\  \Big]\\
    &\iso&
    \End{x,y,z} \Sl{\Set}{I_2\times I_3}\Big[ \Sl{\Set}{I_2\times I_3}\big[\AU{R'}(x),y\tr z\big] \odot \AU{R'} P_1(x) \ ,\ P_2(y)\tr P_3(z)\  \Big]\\
    &\iso&
    \End{x,y,z} \Sl{\Set}{I_3}\Big[ \Sl{\Set}{I_3}\big[\AU{\big.\AU{R'}(x)}(y),z\big] \odot \AU{\big.\AU{R'}
    P_1(x)} P_2(y) \ ,\ P_3(z)\  \Big]\\
    &\iso&
    \End{x,y,z} \Sl{\Set}{I_3}\Big[ \Sl{\Set}{I_3}\big[\AU{R}(x\sq y),z\big] \odot \AU{R} \big(P_1(x) \sq
    P_2(y)\big) \ ,\ P_3(z)\  \Big]\\
    &\iso&
    \End{x,y,z} \Sl{\Set}{I_1\times I_2}\Big[ \Sl{\Set}{I_1\times I_2}\big[x\sq y,\DU{R^\sim}(z)\big] \odot \big(P_1(x) \sq
    P_2(y)\big)\ ,\ \DU{R^\sim} P_3(z)\  \Big]\\
    &\iso&
    \End{z} \Sl{\Set}{I_1\times I_2}\Big[ \Coend{x,y}\Sl{\Set}{I_1\times I_2}\big[x\sq y,\DU{R^\sim}(z)\big] \odot \big(P_1(x) \sq
    P_2(y)\big) \ ,\ \DU{R^\sim}P_3(z)\  \Big]\\
    &\iso&
    \End{z} \Sl{\Set}{I_1\times I_2}\Big[ P_1\Tensor P_2\big(\DU{R^\sim}(z)\big) \ ,\ \DU{R^\sim}P_3(z)\  \Big]\\
    &\iso&
    \Nat\Big( P_1\Tensor P_2 \DU{R^\sim} \ ,\ \DU{R^\sim} P_3\ \Big)\\
    &\iso&
    \Nat\Big( \AU{R} P_1\Tensor P_2 \DU{R^\sim} \ ,\ P_3\ \Big)\\
    &\iso&
    \Nat\Big( \AU{R} P_1\Tensor P_2\ ,\ P_3\AU{R}\ \Big)
    \ .
  \end{eqnarray*}}%
  Because in~$\Set$, all natural transformations are strong, these
  calculations show that there is a natural isomorphism
  between~$\PFSim_\Set[P_1\Tensor P_2,P_3]$
  and~$\PFSim_\Set\big[P_1,\Ran_{\tr}\big(P_2(\BLANK)\tr
  P_3(\BLANK)\big)\big]$.
  Note that because adjoints are unique up-to isomorphisms, the functor we
  just defined is necessarily isomorphic to the one defined on polynomial
  diagrams in~\cite{polyDiagrams}. This implies that~$P_2 \Linear P_2$ is
  indeed well defined and that the Kan extension exists.
\end{proof}
The previous proof relied on the fact that~$\C$ is~$\Set$ in two ways:
\begin{itemize}
  \item strong natural transformations and natural transformations are the
    same thing, so that strong natural transformations can be expressed as an
    end,
  \item $\Set$ has powers and copowers, so that we can use the end / coend
    formulas for~$\BLANK\Linear\BLANK$ and~$\BLANK\Tensor\BLANK$.
\end{itemize}
Proposition~\ref{prop:linearRan} holds for arbitrary~$\C$, but the sequence of
computations needs to be rewritten to use only the universal properties of
left and right Kan extensions, and we need to check that all the natural
isomorphisms respect the strength.



\subsection{Special Case: Polynomial Presheaves and Day's Convolution} 
\label{sub:dayConvolution}

When~$M$ is a small monoidal category, presheaves over~$M$ have a monoidal
structure using Day's convolution product:
\[
  F \Tensor G
  \quad\eqdef\quad
  \Coend{a,b \in M}  M[\BLANK, a\Tensor b] \times F(a)\times G(b)
\]
whenever~$F,G : M^\op \to \Set$. Moreover, this tensor has a right adjoint
making presheaves a symmetric monoidal \emph{closed} category.
The category~$M=\Set^\op$ is monoidal but not small. For~$F,G : \Set \to
\Set$, the formula for Day's
convolution becomes
\[
  F \Tensor G
  \quad\eqdef\quad
  \Coend{a,b \in \Set}  \Set[a\times b,\BLANK] \times F(a)\times G(b)
  \ .
\]
This coend needs not exist as it is indexed by a large category. However,
when~$F$ and~$G$ are polynomial, this is just a special case of
Corollary~\ref{cor:TensorCoend}.
\begin{defi}
  Call a presheaf~$P : \Set \to \Set$ polynomial if it is equivalent to
  \[
    X \quad\mapsto \sum_{a\in A} X^{D(a)}
  \]
  for some set~$A$ and family~$\D:A\to\Set$.
\end{defi}
Corollary~\ref{cor:TensorCoend} implies that polynomial presheaves are closed
under Day's convolution and we have the explicit formula:
\begin{equation}
\label{eqn:DayFormulaPresheaf}
  \Bigg(\sum_{a_1\in A_1} X^{D_1(a_1)}\Bigg)
  \Tensor
  \Bigg(\sum_{a_2\in A_2} X^{D_2(a_2)}\Bigg)
  \quad\iso\quad
  \sum_{(a_1,a_2)\in A_1\times A_2} X^{D_1(a_1)\times D_2(a_2)}
  \ .
\end{equation}
Moreover, the right-adjoint is also polynomial:
\begin{prop}
  The category of polynomial endofunctors on~$\Set$ with Day's convolution
  product is symmetric monoidal closed.
\end{prop}
\begin{proof}
  The category of polynomial endofunctors on~$\Set$ with natural
  transformations between them is a (non full) subcategory of~$\PFSim_\Set$.
  It is thus closed under~$\BLANK\Tensor\BLANK$ and~$\BLANK\Linear\BLANK$.

  To show that~$\Tensor$ and~$\Linear$ are still adjoint in this category, we
  can rewrite the same proof as for Proposition~\ref{prop:linearRan} and
  replace~$R$ everywhere by the trivial span~$\langle1,1\rangle$. The proof
  carries through.

\end{proof}
There is an explicit formula for the right-adjoint~$\BLANK\Linear\BLANK$, but it is much less elegant than the formula for the tensor:
\begin{equation}
\label{eqn:LinearFormulaPresheaf}
    \Bigg(\sum_{a_1\in A_1} X^{D_1(a_1)}\Bigg)
    \Linear
    \Bigg(\sum_{a_2\in A_2} X^{D_2(a_2)}\Bigg)
    \quad\iso\quad
    \sum_{c\in C} X^{E(c)}
\end{equation}
where
  \[
    C \quad\eqdef\quad \sum_{f\in A_1 \to A_2} \ \prod_{a_1\in A_1} D_1(a_1)^{D_2(f(a_1))}
    \qquad\hbox{and}\qquad
    E\big((f,\phi)\big) \quad\eqdef\quad \sum_{a_1\in A_1} D_2\big(f(a_1)\big)
    \ .
  \]
(See~\cite{polyDiagrams}.)



\section{Additive and Exponential Structure} 

\subsection{Additive Structure} 

In~\cite{polyDiagrams}, it was shown that the category of polynomial diagrams
with simulations also has a cartesian / cocartesian structure whenever~$\C$ has
a well behaved coproduct. The coproduct of two diagrams, which is also their
product is defined as:
\[
  P_1\Plus P_2 \quad\eqdef\quad
  \begin{tikzpicture}[baseline=(m-1-1.base)]
    \matrix (m) [matrix of math nodes, column sep=3.5em,text height=1ex, text
    depth=.25ex]
      { I_1+I_2 & D_1+D_2 & A_1+A_2 & J_1+J_2 \\ };
    \path[morphism]
    (m-1-2) edge node[above]{$n_1{+}n_2$} (m-1-1)
    (m-1-2) edge node[above]{$d_1{+}d_2$} (m-1-3)
    (m-1-3) edge node[above]{$a_1{+}a_2$} (m-1-4);
  \end{tikzpicture}
  \ .
\]
A category~$\C$ is \emph{extensive} if the canonical functor~$\Sl{\C}{I_1} \times
\Sl{\C}{I_2} \to \Sl{\C}{I_1+I_2}$ sending~$(f_1,f_2)$ to~$f_1+f_2$ is an equivalence of
category. It implies in particular the following:
\begin{lem}\label{lem:ext_plus}
  If~$\C$ is locally cartesian closed and extensive, we have
  \begin{itemize}
    \item $\Sigma_{f+g}(k+l) \iso \Sigma_f(k) + \Sigma_g(l)$,
    \item $\Delta_{f+g}(k+l) \iso \Delta_f(k) + \Delta_g(l)$,
    \item $\Pi_{f+g}(k+l) \iso \Pi_f(k) + \Pi_g(l)$
  \end{itemize}
  whenever the expressions make sense.
\end{lem}
With that in mind, we have directly that~$\Sem{P_1\Plus P_2}(x+y) \iso
\Sem{P_1}(x)+\Sem{P_2}(y)$ whenever~$x\in\Sl{\C}{I_1}$ and~$y\in\Sl{\C}{I_2}$.
We can thus express the additive structure on
polynomial functors without referring to the underlying polynomials. We have
\begin{lem}\label{lem:plusExtension}
  If~$\C$ is extensive with an initial object~$\Zero$, then:
  \begin{enumerate}
    \item the unique functor from~$\Sl\C\Zero$ to itself is a zero object
      in~$\FSim_\C$,
    \item\label{pt:Plus}
      if we define $F_1\Plus F_2$
      on~$\Sl\C{I_1{+}I_2}\iso\Sl\C{I_1}\times\Sl\C{I_2}$ with
      \[
        F_1{\Plus}F_2 (x+y) \quad\eqdef\quad F_1(x) + F_2(y)
        \ ,
      \]
    then ~$\BLANK\Plus\BLANK$ is a product as well as a coproduct in the
    category~$\FSim_\C$.
  \item $\Zero$ and~$\BLANK\Plus\BLANK$ are a zero object and a
    product/coproduct in~$\PFSim_\C$ as well.
  \end{enumerate}
\end{lem}\noindent
\begin{proof}
  The first point is direct. The second point boils down to the following:
  because~$\C$ is extensive, any span~$I_1+I_2\leftarrow R\to J$ is
  isomorphic to a span~$I_1+I_2\leftarrow R_1+R_2 \to J$ where the left
  leg is~$r_1+r_2$ with~$r_k : I_k\to R_k$ and the right leg is~$[s_1,s_2]$,
  with~$s_k:R_k\to J$. Let's write~$R_k$ for the obvious span~$I_k\leftarrow
  R_k \to J$: its legs are~$r_k$ and~$s_k$.
  Extensivity of~$\C$ implies that:
\begin{itemize}
  \item $\Sem{\AU{R}}(x+y) = [\Sem{\AU{R_1}}(x),\Sem{\AU{R_2}}(y)] : \Sl{\C}J$
    for any~$x+y : \Sl{\C}{I_1+I_2}$;
  \item $\Sem{\DU{R^\sim}}(z) = \Sem{\DU{R_1^\sim}}(z) + \Sem{\DU{R_2^\sim}}(z) :
    \Sl{\C}{I_1+I_2}$ for any~$z:\Sl{\C}{J}$.
(Recall that~$\Sem{\DU{R^\sim}}$ is the right adjoint of~$\Sem{\AU{R}}$ as per
Lemma~\ref{lem:AULeftAdjointLeft}.)
\end{itemize}

We have:
  \begin{eqnarray*}
    &&\Nat\Big(\AU{R} F_1\Plus F_2, G \AU{R}\Big)\\
    &\iso&
    \Nat\Big(\AU{R} F_1\Plus F_2\DU{R^\sim}, G\Big)\\
    &\iso&
    \End{z\in\Sl{\C}{J}} \Sl{\C}{J}\Big[\AU{R}F_1\Plus F_2\DU{R^\sim}(z) , G(z)\Big]\\
    &\iso&
    \End{z} \Sl{\C}{J}\Big[\AU{R}F_1\Plus F_2\Big(\DU{R_1^\sim}(z)+\DU{R_2^\sim}(z)\Big) , G(z)\Big]\\
    &\iso&
    \End{z} \Sl{\C}{J}\Big[\AU{R}\Big(F_1 \DU{R_1^\sim}(z) + F_2\DU{R_2^\sim}(z)\Big) , G(z)\Big]\\
    &\iso&
    \End{z} \Sl{\C}{J}\Big[\AU{R_1}F_1 \DU{R_1^\sim}(z) + \AU{R_2}F_2\DU{R_2^\sim}(z) , G(z)\Big]\\
    &\iso&
    \End{z} \Sl{\C}{J}\Big[\AU{R_1}F_1 \DU{R_1^\sim}(z),G(z)\Big] \times \Sl{\C}{J}\Big[\AU{R_2}F_2\DU{R_2^\sim}(z) , G(z)\Big]\\
    &\iso&
    \End{z} \Sl{\C}{J}\Big[\AU{R_1}F_1 \DU{R_1^\sim}(z),G(z)\Big] \times
    \End{z} \Sl{\C}{J}\Big[\AU{R_2}F_2\DU{R_2^\sim}(z) , G(z)\Big]\\
    &\iso&
    \Nat\Big(\AU{R_1}F_1 \DU{R_1^\sim},G\Big)\times
    \Nat\Big(\AU{R_2}F_2\DU{R_2^\sim}, G\Big)\\
    &\iso&
    \Nat\Big(\AU{R_1}F_1,G\AU{R_1}\Big)\times
    \Nat\Big(\AU{R_2}F_2,G\AU{R_2}\Big)\ .\\
  \end{eqnarray*}
This shows that~$\Plus$ is indeed the coproduct in~$\FSim_\C$. Note that this
proof doesn't rely on the functors~$F_1$ and~$F_2$ being polynomial. The proof
that it is also a product is similar.

To get the last point, i.e., that~$\BLANK\Plus\BLANK$ is also a
coproduct in~$\PFSim_\C$, one needs to show that the natural
isomorphisms given preserves the strengths of natural transformations. This is
left as an exercise...
Another way to prove the last point is simply to use Lemma~\ref{lem:ext_plus}
and the fact that~$\Plus$ is the product and coproduct
in~$\PDSim_\C$~\cite{polyDiagrams}.

\end{proof}

\subsection{Exponential Structure} 

As hinted in~\cite{polyDiagrams}, the category of~$\PDSim_\Set$ has free
commutative~$\Tensor$-comonoids. For a set~$I$, we write~$\Mf(I)$ for the
collection of finite multisets of elements of~$I$. The free
commutative~$\Tensor$-comonoid for~$I\leftarrow D
\rightarrow A \rightarrow I$ is given by
\begin{equation}\label{eqn:freeComonoid}
  \begin{tikzpicture}[baseline=(m-1-1.base)]
    \matrix (m) [matrix of math nodes, column sep=4em, text height=1.5ex, text depth=.25ex]
      {   \Mf(I) & D^* & A^*  & \Mf(I) \\   };
    \path[morphism]
      (m-1-2) edge node[above]{$c_I\circ n^*$} (m-1-1)
      (m-1-2) edge node[above]{$d^*$} (m-1-3)
      (m-1-3) edge node[above]{$c_I\circ a^*$} (m-1-4);
  \end{tikzpicture}
\end{equation}
where
\begin{itemize}
  \item $\BLANK^* : \Set \to \Set$ is the ``list functor'' sending a set~$X$ to the collection of finite
    sequences of elements in~$X$,
  \item $c_I : I^* \to \Mf(I)$ sends a sequence to its equivalence class under
  permutations (multiset).
\end{itemize}
\begin{conj}
  In~$\PFSim_\Set$, the free commutative~$\Tensor$-comonoid over~$F$ is given by
  \[
    !F \quad\eqdef\quad \Sigma_{c_I}\circ F^*\circ\Delta_{c_I}
    \ .
  \]
\end{conj}
This conjecture is a strengthening of the following lemma:
\begin{lem}
  If we write~$\PFSim_{\Set\sim}$ for the category of polynomial functors and
  simulations, where two simulations~$(R,\rho)$ and~$(R',\rho')$ are
  identified when~$R\iso R'$, then:
  \begin{itemize}
    \item $\PFSim_{\Set\sim}$ with~$\Tensor$ and~$\Linear$ is symmetric
      monoidal closed,
    \item $\PFSim_{\Set\sim}$ with~$\Zero$ and~$\Plus$ is cartesian and
      cocartesian,
    \item $\PFSim_{\Set\sim}$ has free commutative~$\Tensor$-comonoids given
      by
  \[
    !F \quad\eqdef\quad \Sigma_{c_I}\circ F^*\circ\Delta_{c_I}
    \ .
  \]
  \end{itemize}
\end{lem}
\begin{proof}
  The first two points follow from earlier results in the paper, and the third
  one is a consequence of the fact that~$(\ref{eqn:freeComonoid})$ is the
  free commutative~$\Tensor$-comonoid
  in~$\PDSim_{\Set\sim}$~\cite{polyDiagrams}. It gives the extensional
  definition of the lemma since it implies that~$\Sem{!P} =
  \Sigma_{c_I}\circ\Sem{P}^* \circ \Delta_{c_I}$.

\end{proof}
It doesn't look too difficult to extend the proof that~$(\ref{eqn:freeComonoid})$
gives the free commutative~$\Tensor$-comonoid in~$\PDSim_{\Set\sim}$
(Proposition~3.2 in~\cite{polyDiagrams}) to the whole of~$\PDSim_\Set$. 
A complete and concise (or at least readable) proof of this fact would be
most welcome and would prove the conjecture...

\subsection{Failure of Classical (Linear) Logic} 

It is natural to ask if the resulting model for intuitionistic linear logic
can be extended to a model for classical linear logic, i.e., if the
monoidal closed structure of the category~$\PFSim_{\Set}$ can be extended to
a~$*$-autonomous structure. The answer is, perhaps unsurprisingly, no.

The full proof isn't very enlightening but let's look at what happens with
the ``natural'' choice of~$\bot \eqdef 1\leftarrow 1\to 1 \to 1$ as a potential
dualizing object.
Take arbitrary objects~$A$ and~$B$ in~$\C$ and define the polynomial
functor~$P_{A,B}(X) = A\times X^B$. As in the case of presheaves on sets
(page~\pageref{eqn:LinearFormulaPresheaf}), there is a simple explicit formula
for~$P_{A,B}^\bot = P_{A,B}\Linear\bot$: we have~$P_{A,B}^\bot(X) =
B^A \times X^A$ and
\[
  P_{A,B}^{\bot\bot}(X) \quad = \quad A^{B^A} X^{B^A}
  \ .
\]
Asking that the canonical simulation from~$P_{A,B}$ to~$P_{A,B}^{\bot\bot}$ is
an isomorphism would imply (by a variant of
Corollary~\ref{cor:isoBetweenPolynomials}) that the canonical map from~$A$
to~\smash{$A^{B^A}$} in~$\C$ is an isomorphism. This is not possible in
general:
\begin{lem}
  Any cartesian closed category~$\C$ in which the canonical natural transformation
  from~$A$ to~\smash{$A^{B^A}$} is an isomorphism is posetal.

  If~$\C$ is also cocartesian, then it becomes a (possibly large)
  ``pre boolean algebra''.
\end{lem}
\begin{proof}
  Because~$\C$ is cartesian closed, we may use simply
  typed~$\fun$-calculus as its internal language. We'll write~$A^B$
  as~$B\Rightarrow A$ as is natural in type theory.
  Taking both~$A$ and~$B$ to be~$C\times C$ above, we have the following
  isomorphism:
  \begin{eqnarray*}
    \C[1,C\times C]
    &\quad\iso\quad&
    \C[1,(C\times C\Rightarrow C\times C) \Rightarrow C\times C]\\
    &\quad\iso\quad&
    \C[C\times C \Rightarrow C\times C, C\times C]
  \end{eqnarray*}
  where each pair~$\langle x,y\rangle$ in~$\C[1,C\times C]$ is sent
  to~$\fun f.\langle x,y\rangle$ in~$\C[C\times C \Rightarrow C\times C, C\times C]$.

  Take~$c_1, c_2$ in~$\C[1,C]$, we have~$\varphi = \fun f. f\,\langle
  c_1,c_2\rangle$ in~$\C[C\times C \Rightarrow C\times C, C\times C]$ which
  satifisfies~$\varphi\,\id = \langle c_1,c_2\rangle$ and~$\varphi\,\tw =
  \langle c_2,c_1\rangle$, where~$\tw$ in~$\C[C\times C, C\times C]$ exchanges
  the left and right components of a pair.
  However, Because of the isomorphism above, $\varphi$ must be of the
  form~$\fun f.\langle x,y\rangle$ for some~$x$ and~$y$ in~$\C[1,C]$. This
  implies that~$\varphi\,\id = \varphi\,\tw$, and thus, that~$c_1=c_2$.
  Because~$\C$ is cartesian closed, any~$f_1,f_2$ in~$\C[A,B]$ correspond
  precisely to constants~$\ulcorner\!f_1\!\urcorner$ and~$\ulcorner\!f_2\!\urcorner$
  in~$\C[1,A\Rightarrow B]$ and are thus equal. The category~$\C$ is thus posetal.

  Because~$\C$ is cartesian closed, it is thus a (possibly large) Heyting
  semi-lattice; and if~$\C$ is also cocartesian, it becomes a Heyting algebra.
  Now, the existence of a transformation from~$(A\Rightarrow B)\Rightarrow A$
  to~$A$ amounts to saying the law of Peirce is satisfied. This makes the
  Heyting algebra boolean...

\end{proof}

Note however that the model of predicate transformers from~\cite{denotPT}
can be seen as a ``proof irrelevant'' variant of~$\PFSim_\Set$:
\begin{itemize}
  \item we collapse the categories~$\Sl{\Set}{I}$ into preorders, making each~$\Sl{\Set}{I}$
    equivalent to the algebra of subsets of~$I$,
  \item we collapse the categories of spans into preorders, making each~$\Span[I,J]$
    equivalent to the algebra of relations between~$I$ and~$J$,
  \item we identify simulations~$(R,\alpha)$ and~$(R,\beta)$.
\end{itemize}
This gives a non-trivial~$*$-autonomous category:
\begin{itemize}
  \item polynomial functors on~$\Sl{\Set}{I}$ become monotonic transformations
    on~$\Pow(I)$, all of which are in fact ``polynomial'',
  \item simulations become relations satisfying a closure property,
  \item duality gives~$P^\bot(x) = \overline{P (\overline{x})}$ whenever~$P
  : \Pow(I) \to \Pow(I)$, where~$\overline y$ is the complement of~$y$ with
  respect to~$I$.
  \item the dualizing object is~$\One$, the unit of the tensor, i.e.,
    $\One=\bot$.
  \item However, the category is not compact closed as~$\BLANK\Tensor\BLANK$ is
    different from its dual.
\end{itemize}


\subsection{Future Work} 

The first question that comes to mind is what happens if we consider analytic
functors on~$\Set$, i.e., those of the form
\[
  X \quad\mapsto \sum_{a\in A} X^{D(a)}/_{G(a)}
\]
where each~$G(a)$ is a subgroup of the automorphisms of~$D(a)$, acting in an
obvious way on~$X^{D(a)}$. Formula~\ref{eqn:DayFormulaPresheaf} from
Section~\ref{sub:dayConvolution} has a natural generalization to this context.
Does it work as an intensional formula for Day's convolution? What about the
analogous to formula~\ref{eqn:LinearFormulaPresheaf}?

\section{Thanks}

The author would like to thank both Neil Ghani and Paul-Andr\'e Melli\`es for
helpful discussions on early versions of this work.

\bibliographystyle{amsalpha}
\bibliography{poly}


\newpage
\appendix


\section{Locally Cartesian Closed Categories} 
\label{app:LCCC}

For a category~$\C$ with finite limits, we write~``$\One$'' for its terminal
object and~``$A\times B$'' for the cartesian product of~$A$ and~$B$.
The ``pairing'' of~$f:C\to A$ and~$g:C\to B$ is written~$\langle f,g\rangle :
C\to A\times B$.

If~$f:A\to B$ is a morphism, it induces a pullback functor~$\Delta_f$ from
slices over~$B$ to slices over~$A$. This functor has a left adjoint~$\Sigma_f$
which is ``pre-composition by~$f$''. When all the~$\Delta_f$s also
have a right adjoint, we say that~$\C$ is \emph{locally cartesian closed}.
The right adjoint is written~$\Pi_f$. We thus have
\[
  \Sigma_f \quad\dashv\quad \Delta_f \quad\dashv\quad \Pi_f \quad .
\]

Besides the isomorphisms coming from the adjunctions, slices enjoy two
fundamental properties:
\begin{itemize}
  \item \emph{the Beck-Chevalley isomorphisms:}
    \[
      \Pi_g \, \Delta_l  \quad\iso\quad  \Delta_k \, \Pi_f
      \qquad\hbox{and}\qquad
      \Sigma_g \, \Delta_l  \quad\iso\quad  \Delta_k \, \Sigma_f
    \]
    whenever
    \[
      \begin{tikzpicture}[baseline=(m-1-1.base)]
        \matrix (m) [matrix of math nodes, row sep={4em,between origins},
        column sep={4em,between origins},text height=1ex, text depth=.25ex]
          { \cdot & \cdot \\ \cdot & \cdot \\ };
        \path[morphism]
          (m-1-1) edge node[auto]{$g$} (m-1-2)
          (m-1-1) edge node[left]{$l$} (m-2-1)
          (m-2-1) edge node[below]{$f$} (m-2-2)
          (m-1-2) edge node[auto]{$k$} (m-2-2);
        \pullback[.5cm]{m-2-1}{m-1-1}{m-1-2}{draw,-};
      \end{tikzpicture}
    \]
    is a pullback,

  \item \emph{distributivity:} when $b:C\to B$ and $a:B\to
    A$, we have a commuting diagram
    \begin{equation}\label{diag:distr}
      \begin{tikzpicture}[baseline=(m-2-1.base)]
        \matrix (m) [matrix of math nodes, row sep={2.5em,between origins},
        column sep={2.5em,between origins},text height=1ex, text depth=.25ex]
          {   & \cdot && \cdot \\
            C &       &&       \\
              &  B    && A     \\};
        \path[morphism]
          (m-1-2) edge node[auto]{$a'$} (m-1-4)
          (m-1-2) edge (m-3-2)
          (m-1-2) edge node[over]{$\epsilon$} (m-2-1)
          (m-2-1) edge node[below left]{$b$} (m-3-2)
          (m-3-2) edge node[below]{$a$} (m-3-4)
          (m-1-2) edge node[left]{$u'$} (m-3-2)
          (m-1-4) edge node[auto]{$u \eqdef \Pi_a(b)$} (m-3-4);
        \pullback[.5cm]{m-3-2}{m-1-2}{m-1-4}{draw,-};
      \end{tikzpicture}
    \end{equation}
    where~$\epsilon$ is the co-unit of~$\Delta_a\dashv \Pi_a$. For such a
    diagram, we have
    \[
      \Pi_a \, \Sigma_b \quad\iso\quad \Sigma_{u} \, \Pi_{a'} \, \Delta_\epsilon
      \ .
    \]

\end{itemize}



\section{Dependent Type Theory} 
\label{app:DTT}

In~\cite{Seely}, Seely showed how an extensional version of Martin L\"of's
theory of dependent types~\cite{ML84} could be regarded as the internal
language for locally cartesian closed categories. A little later, Hofmann
showed in~\cite{Hofmann} that Seely's interpretation works only ``up-to
canonical isomorphisms'' and proposed a solution.

A type~$A$ in context~$\Gamma$, written~``$\Gamma\vdash A\ \mathsf{type}$'' is interpreted as a
morphism~$a:\Gamma_{\!A}\to\Gamma$, that is as an object in the slice over
(the interpretation of)~$\Gamma$. Then, a term of type~$A$ in context~$\Gamma$,
written~``$\Gamma\vdash t:A$'' is interpreted as a
morphism~$u:\Gamma\to\Gamma_{\!A}$ such that~$a u=1$, i.e., a section of (the
interpretation of) its type.  When~$A$ is a type in context~$\Gamma$, we
usually write~$A(\gamma)$ to emphasize the dependency on the context and we
silently omit irrelevant parameters.
If we write~$\inter{a}{\Gamma\vdash A}$ to mean that the interpretation
of type~$\Gamma\vdash A$ is~$a$, the main points of the Seely semantics are:
\[
  \Rule
  {\inter{a}{\Gamma \vdash A} \qquad \inter{b}{\Gamma,x:A\vdash B(x)}}
  {\inter{\Pi_a(b)}{\Gamma \vdash \prod_{x:A} B(x)}}
  {product}
  \ ,
\]
\[
  \Rule
  {\inter{a}{\Gamma \vdash A} \qquad \inter{b}{\Gamma,x:A\vdash B(x)}}
  {\inter{\Sigma_a(b)}{\Gamma \vdash \sum_{x:A} B(x)}}
  {sum}
  \ ,
\]
\[
  \Rule
  {\inter{f}{\Gamma \vdash \vec u:\Delta} \qquad \inter{u}{\Delta\vdash
  A(\vec x)}}
  {\inter{\Delta_f(u)}{\Gamma \vdash A(\vec u)}}
  {substitution}
  \ .
\]
Of particular importance is the distributivity condition~$(\ref{diag:distr})$
whose type theoretic version is an intensional version of the axiom of choice:
\begin{equation}\label{diag:AC}
  \Gamma \ \vdash \ \prod_{x:A} \ \sum_{y:B(x)} U(x,y)
  \qquad\iso\qquad
  \Gamma \ \vdash \ \sum_{f:\prod_{x:A} B(x)} \ \prod_{x:A} \ U\big(x,f(x)\big)
  \ .
\end{equation}

Extensional type theory has a special type for equality ``proofs''. Its
formation and introduction rules are as follows:
\[
  \Rule
  {\Gamma \vdash u_1 : A \qquad \Gamma \vdash u_2 : A}
  {\Gamma \vdash \IdX{A}{u_1}{u_2}}
  {}
\qquad
  \Rule
  {\Gamma \vdash u:A}
  {\Gamma \vdash \refl_A(u) : \IdX{A}{u}{u}}
  {}
  \ .
\]
This type for equality is ``extensional'' because having an inhabitant
of~$\IdX{X}{u}{v}$ implies that~$u$ and~$v$ are definitionally (extensionally) equal. This
is reflected in its interpretation:
\[
  \Rule
  {\inter{f_1}{\Gamma \vdash u_1:A} \qquad \inter{f_2}{\Gamma \vdash u_2:A}}
  {\inter{\eq(f_1,f_2)}{\Gamma \vdash \IdX{X}{x_1}{x_2}}}
  {identity type}
\]
where~$\eq(f_1,f_2)$ is the equalizer of~$f_1$ and~$f_2$. (This is indeed a
slice over the interpretation of~$\Gamma$.)

\pdfinfo{
  /Title (A Linear Category of Polynomial Functors -- II)
  /Author (Pierre Hyvernat)
  /Subject (Category Theory)
  /Keywords (polynomial functors; category theory; linear logic)}

\end{document}